\documentclass[twoside]{article}
\usepackage[letterpaper,margin=1in]{geometry}
\usepackage{enumitem}
\usepackage{amsthm}
\usepackage{tikz}
\usetikzlibrary{decorations.pathreplacing,calligraphy}
\usepackage{amsmath}
\usepackage[charter]{mathdesign}
\usepackage[linktocpage=true,pagebackref=true,colorlinks,linkcolor=magenta,citecolor=blue,bookmarks,bookmarksopen,bookmarksnumbered]{hyperref}
\usepackage[boxruled,vlined,linesnumbered]{algorithm2e}
\SetKwRepeat{Do}{do}{while}
\usepackage{comment}
\usepackage{lineno}
\usepackage[nameinlink]{cleveref}
\usepackage{eulervm}
\usepackage{booktabs} 
\usepackage{multirow}
\usepackage{quoting}
\usepackage{colortbl}
\usepackage{xcolor}

\newtheorem{theorem}{Theorem}[section]
\newtheorem{lemma}[theorem]{Lemma}

\newtheorem{definition}[theorem]{Definition}
\newtheorem{remark}[theorem]{Remark}
\newtheorem{obs}[theorem]{Observation}

\newtheorem{assumption}[theorem]{Assumption}

\newcommand{\TL}{\widetilde{O}}
\newcommand{\todo}[1]{[{\color{red}{#1}}]}
\newcommand{\CON}{\odot}

\newcommand{\pix}{pivot_i(x)}

\newcommand{\QU}{\textsc{Query}}
\newcommand{\pis}{pivot_i(s)}
\newcommand{\piis}{pivot_{i+1}(s)}
\newcommand{\pit}{pivot_i(t)}

\newcommand{\al}{\alpha}
\newcommand{\be}{\beta}
\newcommand{\ball}{ball}
\newcommand{\ti}[1]{2^{2^{#1}}}
\newcommand{\uf}{\textsc{UpdateFrom}}
\newcommand{\wtst}{\left\lceil \frac{|st|}{2}\right\rceil}
\newcommand{\lgn}{\log\log n}
\newcommand{\lln}{\log\log n-1}
\newcommand{\llnt}{\log\log n-2}
\newcommand{\llni}{\log \log n-i}
\newcommand{\llnk}{\log \log n-\log k}
\newcommand{\ui}{u_i}
\newcommand{\vi}{v_i}
\newcommand{\ai}{a_i}
\newcommand{\bi}{b_i}

\newcommand{\uii}{u_{i+1}}
\newcommand{\vii}{v_{i+1}}
\newcommand{\aii}{a_{i+1}}
\newcommand{\bii}{b_{i+1}}

\newcommand{\tri}{\textsc{Triangulate}}

\newcommand{\ra}{}
\newcommand{\es}{\textsc{EnsureCloseness}}

\title{Improved 2-Approximate Shortest Paths for close vertex pairs 
}
\date{}
\author{
  Manoj Gupta \\
  IIT Gandhinagar \\
  India\\
  \texttt{gmanoj@iitgn.ac.in} \\
}
\begin{document}
\setlength{\abovedisplayskip}{1pt}
\setlength{\belowdisplayskip}{5pt}

\maketitle

\begin{abstract}

An influential result by Dor, Halperin, and Zwick (FOCS 1996, SICOMP 2000) implies an algorithm that can compute approximate shortest paths for all vertex pairs in $\TL(n^{2+O\left(\frac{1}{k}\right )})$ time\footnote{$\TL()$ notation hides poly$\log n$ factors}, ensuring that the output distance is at most twice the actual shortest path, provided the pairs are at least $k$ apart, where $k \ge  2$. We present the first improvement on this result in over 25 years. Our algorithm achieves roughly same $\TL(n^{2+\frac{1}{k}})$ runtime but applies to vertex pairs merely $O(\log k)$ apart, where $\log k \ge 1$. When $k=\log n$,  the running time of our algorithm is $\TL(n^2)$ and it works for all pairs at least $O(\lgn)$ apart. Our algorithm  is combinatorial, randomized, and returns correct results for all pairs with a high probability.

\end{abstract}

\section{Introduction}

Let $G$ be an undirected, unweighted graph with $n$ vertices and $m$ edges.  Let $st$ denote the shortest path between $s$ and $t$ and  $|st|$ denote its length. Computing all-pairs shortest paths (APSP) is a fundamental problem in graph algorithm. For unweighted graphs, combinatorial methods achieve a time complexity of $O(mn)$, while fast matrix multiplication ($FMM$) solves APSP in $\TL(n^{\omega})$ time \cite{Seidel95}, where $\omega$ denotes the matrix multiplication exponent \cite{AlmanW24,Gall14,Stothers10,Williams12,DuanHR23,WilliamsXZR24}.

However, these algorithms have super-quadratic runtimes in dense graphs. To achieve near-quadratic performance, we must allow approximate solutions. Let $est(s,t)$ be the approximate shortest path length computed by an algorithm for each pair $(s,t)$. An algorithm is said to be an $(\alpha,\beta)$-approximate APSP, or to provide an $(\alpha,\beta)$-approximation of the length of a path, if for every pair $(s,t)$, the following holds:
\[
|st| \leq est(s,t) \leq \alpha |st| + \beta \hspace{2.5cm} \text{(where $\al \ge 1$ and $\be \ge 0$)}
\]

If $\beta = 0$, we omit it from the notation and call the algorithm $\alpha$-approximate APSP. When $\alpha = 1$, we refer to it as a $+\beta$-approximate APSP. The former provides a multiplicative approximation, while the latter gives an additive approximation. We first discuss prior results on additive approximation.

Aingworth, Chekuri, Indyk and Motwani \cite{AingworthCIM99} designed an $\TL(n^{2.5})$ algorithm for $+2$-approximate  APSP. Dor, Halperin and Zwick \cite{DorHZ00} improved the above time to $\TL(\min\{n^{3/2}m^{1/2},n^{7/3}\})$. Deng, Kirkpatrick, Rong, Williams and Zhong \cite{DengYRWZ22}  used $FMM$ to improve this running time to $\TL(n^{2.26})$(using an improved result in \cite{Durr23}).

  For a $+k$-approximate APSP, Dor, Halperin, and Zwick \cite{DorHZ00} designed a combinatorial algorithm with a running time of $\TL(\min\{n^{2-\frac{2}{k+2}}m^{\frac{2}{k+2}},n^{2+\frac{2}{3k-2}}\})$. Saha and Ye \cite{SahaY24} recently improved this bound using $FMM$. These results serve as fundamental building blocks for multiplicative approximation.

   In this paper, we primarily focus on  2-approximate APSP algorithm. A 2-approximate APSP algorithm can be constructed using a $+k$-approximate APSP algorithm based on the following observation:

\begin{obs}\label{obs:atom}  
A $+k$-approximate APSP algorithm provides $2$-approximate distances for all vertex pairs with a distance of at least $k$.  
\end{obs}

 Using this observation, the above algorithms produce 2-approximate answers for vertex pairs sufficiently far apart.
 Thus, +2-approximate algorithm of  \cite{AingworthCIM99} implies a 2-approximate APSP algorithm with a running time $\TL(n^{2.5})$. In particular, look at the running time of Dor, Halperin and Zwick \cite{DorHZ00}. If $k = O(\log n)$, their algorithm's running time is $\TL(n^2)$.  This result implies that we can compute a 2-approximate APSP for vertex pairs at least $O(\log n)$ apart.   One of the most important open questions in this field is {\em designing a 2-approximate APSP algorithm that runs in $\TL(n^2)$ time and handles all vertex pairs}.

 Unfortunately, thus far, only $(2, \beta)$-approximate APSP algorithms have been achieved in $\TL(n^2)$ time, where $\beta$ is a small constant. Indeed, several algorithms get $(2,1)$-approximate APSP in $\TL(n^2)$ time \cite{BermanK07, PatrascuR10, Sommer16, Knudsen17}.    For an exact 2-approximate APSP, Cohen and Zwick \cite{CohenZ01} proposed an $\TL(n^{3/2}m^{1/2})$  time algorithm, while Baswana and Kavitha \cite{BaswanaK06} achieved $\TL(m\sqrt{n} + n^2)$ time.  Both perform well on sparse graphs, but \cite{BaswanaK06} is more efficient than \cite{CohenZ01}. In fact, Baswana and Kavitha's algorithm serves as a building block for many future algorithms.  
  Until recently, these remained the best-known results for 2-approximate APSP. However, recent advancements have improved upon them.

Roditty \cite{Roditty23} designed a combinatorial algorithm for 2-approximate APSP in $\TL(n^{2.25})$ time, reviving interest in the problem. Dory, Forster, Kirkpatrick, Nazari, Williams, and Vos \cite{DoryFKNWV24}, along with Saha and Ye \cite{SahaY24}, leveraged $FMM$ to improve this bound to $\TL(n^{2.031})$. See \Cref{tab:comparison} for relevant results on 2-approximate APSP in undirected and unweighted graphs.

\begin{table}[h]
    \centering
    \begin{tabular}{p{2.5cm}|p{2.5cm}|p{5cm}|c}
        \toprule
        \small{Running Time} & \small{Applicable for pairs with distance at least } & \small{Remark} & \small{Reference} \\
        \midrule
        $\TL(n^{2.5})$ & 1 &  & \cite{AingworthCIM99} \\
        \midrule
        $\TL(n^{2+O(1/k)})$ & $k$ & $k \ge 2$ is an integer & \cite{DorHZ00,Roditty23,SahaY24} \\
        \midrule
        $\TL(n^{2})$ & $\log n$ & Consequence of above result & \cite{DorHZ00,Roditty23,SahaY24} \\
        \midrule
        $\TL(n^{3/2}m^{1/2})$ &1&Better for sparse graphs&\cite{CohenZ01}\\
        \midrule
        $\TL(m\sqrt n + n^2)$ & 1 & Better for sparse graphs & \cite{BaswanaK06} \\
        \midrule
        $\TL(n^{2.25})$ & 1 & & \cite{Roditty23} \\
        \midrule
        $\TL(n^{2.031})$ & 1 & Uses $FMM$ & \cite{DoryFKNWV24,SahaY24} \\
        \midrule
        \rowcolor{lightgray}
        $\TL(n^{2+1/k})$ & $O(\log k)$ & $\log k \ge 1$ is an integer & \Cref{thm:main} \\

        \midrule
        \rowcolor{lightgray}
        $\TL(n^{2})$ & $O(\log \log n)$ & Consequence of the above result &\\
    \end{tabular}
    \caption{Some important results for 2-approximate APSP }
    \label{tab:comparison}
\end{table}

But what about computing 2-approximate APSP for vertex pair at least $k$ apart? To this end, we use \Cref{obs:atom}  and apply the additive APSP result of \cite{DorHZ00,SahaY24}. The drawback of this approach is using a powerful result for a seemingly simpler problem. The algorithm in \cite{DorHZ00,SahaY24} actually computes a $+k$-approximate APSP, which incidentally implies a 2-approximate APSP for pairs at least $k$ apart. Intuitively, solving $+k$-approximate APSP appears more challenging --- or is it? Roditty \cite{Roditty23} explicitly raises this question. We quote it here (with minor notation changes):

\begin{quoting}[font=itshape, begintext={``}, endtext={''\cite{Roditty23}}]  
Question 1.3. Can we compute $2$-approximation faster than computing $+k$ approximation for vertex pairs at distance at least $k$?  
\end{quoting}

Roditty \cite{Roditty23} affirmatively answered this question by designing a $2$-approximate APSP algorithm with a running time of $\TL(\min\{n^{2-\frac{2}{k+4}}m^{\frac{2}{k+4}},n^{2+\frac{2}{3k-2}}\})$ for vertex pairs at least $k$ apart. Saha and Ye \cite{SahaY24} further improved this bound to $\TL(\min\{n^{2-\frac{2}{k+4}}m^{\frac{2}{k+4}},n^{2+\frac{1}{2(k-1)}}, n^{2+\frac{2}{3k-2}}\})$. These results refine the original work of \cite{DorHZ00} and establish a separation between $+k$-approximate and $2$-approximate APSP for such vertex pairs.  However, all three results \cite{DorHZ00,Roditty23,SahaY24} still have a maximum running time of $O(n^{2+O(1/k)})$. In this paper, we make significant progress on this problem, establishing the following result:

\begin{theorem} \label{thm:main}
	There is a randomized combinatorial algorithm that, with high probability, finds $2$-approximate shortest paths for all vertex pairs at least $O(\log k)$ apart and runs in $\TL(n^{2+\frac{1}{k}})$ time, where $\log k \ge  1$ is an integer.
\end{theorem}

Thus, we can compute a $2$-approximation significantly faster than obtaining a $+k$-approximation for vertex pairs at least $k$ apart. Specifically, when $k = \log n$, the running time of our algorithm is $\TL(n^2)$ and it applies to all pairs at least $O(\lgn)$ apart. In fact, most of our paper focuses on this special case -- our starting point for this problem. Once we resolve this case, we extend our algorithm to any $k$. Our result is randomized and relies on a simple combinatorial technique that, to our knowledge, has not been exploited for this problem. We hope this technique will prove useful in related problems within this field.

\subsection{Other related results}
In the literature, extensive work has been done on the $(\alpha, \beta)$-APSP problem. We now highlight some major threads in this area.

\begin{enumerate}
	\item  $(\al,\be)$-approximate APSP in weighted graphs
	
	The algorithms of \cite{CohenZ01} and \cite{BaswanaK06} (described in the previous section) extend to weighted graphs. Kavitha \cite{Kavitha12} designed a $(2+\epsilon)$-approximate APSP for weighted graphs in $\TL(n^{2.25})$ time. \cite{CohenZ01, BaswanaK06} designed a $3$-approximate APSP in $\tilde{O}(n^2)$ time. Additionally, \cite{BaswanaK06} designed a $(2, W_{s,t})$-approximate APSP in $\tilde{O}(n^2)$ time, where $W_{s,t}$ is the maximum edge weight in the shortest path between $s$ and $t$. 

	\item Approximate Distance Oracle

	Given a graph $G$, we can preprocess it to construct a data structure that supports approximate distance queries between any two vertices. To this end, we design a query algorithm that returns an estimate upon request. We extend the definition of $(\alpha, \beta)$-approximate APSP to an $(\alpha, \beta)$-approximate distance oracle, where distances are computed only when queried rather than for all vertex pairs. Consequently, the preprocessing time can be subquadratic when $m < n^2$.

	In their seminal paper, Thorup and Zwick \cite{ThorupZ01} designed a $(2k-1)$-approximate distance oracle with a preprocessing time of $O(kmn^{1/k})$, oracle size $\TL(kn^{1+1/k})$, and query time $O(k)$. This result was later improved in \cite{BaswanaK06, Wulff12, Wulff13, Chechik14, Chechik15}. Akav and Roditty \cite{AkovR20} developed a $(2+\epsilon, 5)$-approximate distance oracle with a preprocessing time of $\TL(m+n^{2-\Omega(\epsilon)})$, where $\epsilon \in (0,1/2)$ is a constant. More recently, Chechik and Zhang \cite{ChechikZ22} designed a $(2,3)$-approximate distance oracle of size $\TL(n^{5/3})$ with a preprocessing time of $\TL(m+n^{1.987})$. They also constructed a $(2+\epsilon, c(\epsilon))$-approximate distance oracle of size $\tilde{O}(n^{5/3})$ in $\tilde{O}(m + n^{5/3+\epsilon})$ time, where $c(\epsilon)$ depends exponentially on $1/\epsilon$.

\end{enumerate}

\section{Preliminaries}

Unless stated otherwise, $st$ represents the shortest path between $s$ and $t$ in $G$, with ties broken arbitrarily. Let $|st|$ denote the length of $st$ path.
Although the input graph is unweighted, our algorithm constructs multiple graphs which may have weighted as well as unweighted edges. If the edge from $s$ to $t$ is an unweighted edge in our constructed graph, we will denote it as $(s,t)$, else we will denote it as $[s,t]$. Abusing notation, the weight of the edge $[s,t]$ will be denoted as $|[s,t]|$.  Let $\deg(s)$ be the degree of $s$ in $G$. For an edge $e=(s,t)$ in $G$, define the degree of  $e$, or  $\deg(e) := \min\{\deg(s),\deg(t)\}$. For two path $P$ and $Q$ ending and starting at same vertex respectively, $P \odot Q$ denotes the concatenation of $P$ with $Q$. In our paper, the term {\em "with a high probability"} means a probability $\ge 1-\frac{1}{n^c}$ where $c >0$ is a constant. Throughout the paper, for clear presentation, we will assume that $\lgn$ is an integer.

 \begin{definition} (Estimate of a pair)
 	For each pair $(s,t)$ in the graph, we maintain $est(s,t)$ which contains the estimate of the length of $st$ path. Several algorithms in the literature compute a $(2,1)$-approximate shortest path for all $(s,t)$ pairs in $\TL(n^2)$ time~\cite{BermanK07,PatrascuR10,Sommer16,Knudsen17}. Henceforth, for each $(s,t)$ pair in the graph, we will assume that: $est(s,t) \le 2|st|+1$. Additionally, whenever our algorithm updates $est(s,t)$, it also updates $est(t,s)$, ensuring both always store the same value.

.
 \end{definition}

\begin{definition} (Nested Vertex Sampling)
Let $A_{0}$ be the set of all  vertices in the graph. Let $A_1$ be the set of vertices sampled with the probability $\frac{1}{\ti{1}} $ from the graph. For each $ 2\le i \le \lln$, $A_i$ is the set of vertices sampled with a probability $\frac{1}{\ti{i-1}}$ from $A_{i-1}$. Our sampling ensures that $A_{\lln} \subseteq A_{\lgn-2} \subseteq \dots \subseteq A_2 \subseteq A_1 \subseteq A_0$.
\end{definition}

The following lemma follows directly from the definition and is proved in the appendix.

\begin{lemma}\label{lem:prob}
	$Pr[v \in A_i] = \frac{1}{\ti{i}}$ for $i \in [1 \dots \lln]$. Using standard probabilistic arguments, the size of $A_i$ is $\TL\left(\frac{n}{\ti{i}}\right)$ with a high probability.
\end{lemma}

\begin{definition} (Pivots and Balls)

	For $0 \le i \le \lln$, let $i$-th pivot of $s$ or $\pis$ be the vertex of $A_i$ nearest to $s$, with ties broken arbitrarily. Formally,
	$$\pis = \arg\min_{x \in A_i}\{ |sx| \}$$
	Define $\ball_i(s)$ to be all the vertices that are closer to $s$ than $\pis$. Formally, $$\ball_i(s) = \{v \in G  \text{ such that } |sv| < |s~\pis| \}$$
\end{definition}

By definition, $ball_i(s)$ cannot contain any vertex from $A_i$. The following lemma is immediate using the above definitions and we prove it in the appendix.

\begin{lemma} \label{lem:ballsize} (Size of ball)
With high probability, for all $0 \leq i \leq \lln$ and $s \in G$, the size of $ball_{i}(s)$ is at most $\TL(\ti{i})$.  

\end{lemma}

 We can use \Cref{lem:ballsize} to show that vertices  strictly inside the ball have bounded degree. $\ball_{i}(s)$ contains all vertices at a distance less than $|s \pis|$ from $s$.  Consider any vertex $v$ not on the boundary, i.e., $|sv| \leq |s~\pis| - 2$.  We claim that its degree is at most $\TL(\ti{i})$.  Otherwise, each of its neighbors would belong to $\ball_{i}(s)$, contradicting \Cref{lem:ballsize}.  Thus, we establish the following lemma.  

\begin{lemma}\label{lem:degreeinball}

With high probability, each vertex $v \in \ball_{i}(s)$ satisfying $|sv| \leq |s~\pis| - 2$ has degree at most $\TL(\ti{i})$.  

\end{lemma}

Lemma~\ref{lem:degreeinball} implies that the shortest path from $s$ to $\pis$ consists of edges with degree at most  $\TL(\ti{i})$, except possibly the last edge.  We  use this to compute $\pis$ and $\ball_i(s)$ efficiently for each  $s \in G$ and each $i$.  The following lemma is proved in the appendix.

\begin{lemma}\label{lem:pivotandball}
With high probability, for each $i \in [0, \lln]$ and $s \in G$, we can compute $\pis$ and $\ball_i(s)$ in $\TL(n^2)$ time.
\end{lemma}

Henceforth, we will assume  that for each $s\in G$ and for each $i$, $est(s,\pis) = |s~\pis|$. Also, for each vertex $w \in \ball_i(s)$, $est(s,w) = |sw|$. All the above definitions and lemmas are well-known in the literature. Hence, we have relegated their proofs to the appendix. Before introducing new notations, we apply a key result.  Baswana and Kavitha~\cite{BaswanaK06} proposed an $\TL(m\sqrt{n})$ algorithm for  2-approximate APSP. We use it to approximate shortest paths where every edge in the path has degree at most $\TL(\sqrt{n})$.  To achieve this, we construct a subgraph $H$ containing all edges with degree at most $\TL(\sqrt{n})$ and run the algorithm from~\cite{BaswanaK06} on it.  Since $H$ has size $\TL(n\sqrt{n})$, the algorithm runs in $\TL(n^2)$ time.  Thus, we only need to handle $st$ paths that contain a vertex with degree $\ge \TL(\sqrt{n})$.  Hence, we assume the following for the rest of the paper.

\begin{assumption} \label{eq:1}
The shortest $st$ path contains at least one vertex with degree $\geq \TL(\sqrt{n})$. 
\end{assumption}  

\subsection{New Notations}

\begin{definition}

For $0\le i \le \lln$, let $\ai$ be the vertex closest to $s$ on the $st$ path such that its distance to $pivot_i(\ai)$ is at most $1$. We set $\ui := pivot_i(\ai)$.  

Similarly, let $\bi$ be the vertex closest to $t$ on the $st$ path such that its distance to $pivot_i(\bi)$ is at most $1$, and set $\vi := pivot_i(\bi)$.  
See \Cref{fig:1} for an example.

\end{definition}

\begin{figure}[hpt!]
\begin{center}
\begin{tikzpicture}
    \node[draw, circle, fill=black, label=below:$s$, minimum size=4pt, inner sep=0pt] (s) at (0,0) {};
    \node[draw, circle, fill=black, label=below:$t$, minimum size=4pt, inner sep=0pt] (t) at (12,0) {};

    \draw[thick] (s) -- ++(2,0) coordinate (p1)
                  -- ++(2,0) coordinate (p2)
                  -- ++(2,0) coordinate (p3)
                  -- (t);

    \node[draw, circle, fill=blue, label=above:$\ui$, minimum size=4pt, inner sep=0pt] (ui) at (2,1) {};
    \node[draw, circle, fill=blue, label=above:$\uii$, minimum size=4pt, inner sep=0pt] (uii) at (4,1) {};
    \node[draw, circle, fill=red, label=below:$\aii$, minimum size=4pt, inner sep=0pt] (aii) at (3.5,0) {};
    \node[draw, circle, fill=red, label=below:$\ai$, minimum size=4pt, inner sep=0pt] (ai) at (2,0) {};
    \node[draw, circle, fill=blue, label=right:$\vi$, minimum size=4pt, inner sep=0pt] (vi) at (9,-1) {};
    \node[draw, circle, fill=red, label=above:$\bi$, minimum size=4pt, inner sep=0pt] (bi) at (9,0) {};

    \node[draw, circle, fill=blue, label=right:$\vii$, minimum size=4pt, inner sep=0pt] (vii) at (7,-1) {};
    \node[draw, circle, fill=red, label=above:$\bii$, minimum size=4pt, inner sep=0pt] (bii) at (8,0) {};

    \draw[thick] (ui) -- (ai);
    \draw[thick] (vi) -- (bi);
    \draw[thick] (vii) -- (bii);
    \draw[thick] (aii) -- (uii);
\end{tikzpicture}
\end{center}
\caption{A pictorial view of $\ui$, $\ai$, $\vi$, $\bi$. }
\label{fig:1}
\end{figure}
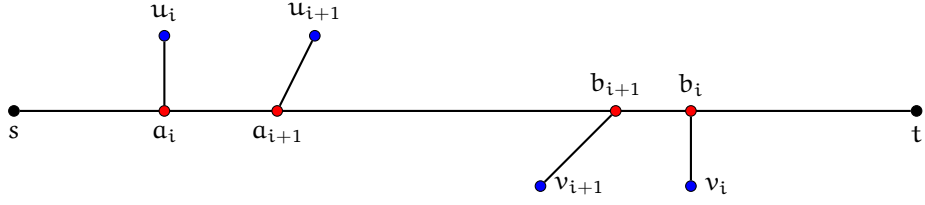

By definition, $a_0 = u_0=s$ and $b_0=v_0=t$. Note that  $\ui$ is set to $pivot_i(\ai)$ which by definition is either $\ai$ or adjacent to $\ai$.  In general, the existence of $\ai$, $\ui$, $\bi$ and  $\vi$ is not always guaranteed. To address this, we rely on \Cref{eq:1}. Suppose the $st$ path contains a vertex $x$ with degree at least $\TL(\sqrt{n})$. With high probability, some vertex in $A_{\lln}$ is adjacent to $x$, ensuring the existence of $u_{\lln}$.  Since vertex sampling is nested, this also guarantees the presence of $u_{\lgn-2}, \dots, u_1, u_0$ along $st$. The same holds for $\vi$'s.   Nested vertex sampling directly implies that $\ai$ is closer to $s$ than $\aii$. Similarly, $\bi$ is closer to $t$ than $\bii$.  Thus, we claim the following lemma which directly follows from the fact that $A_{i+1} \subseteq A_{i}$.

\begin{lemma}
For all $0 \le i \le \lgn-2$, $|s\ai| \le |s\aii|$ and $|t\bi| \le |t\bii|$.
\end{lemma}

Also, note that there cannot be any vertex with degree $\TL(\ti{i})$ on the $s\ai$ path, except maybe $\ai$. If such a vertex $x$ exists, then with a high probability there is a vertex of $A_i$ adjacent to $x$, contradicting our choice of $\ui$ and $\ai$. Thus, we have:

\begin{lemma} \label{lem:onpath}
	With high probability, for all $0 \le i \le \lln$, each edge in the  $s\ai$ path has degree $\le \TL(\ti{i})$. Similarly, each edge in $t\bi$ has degree $\le \TL(\ti{i})$.
\end{lemma}

\section{Overview}
\label{sec:overview}
In majority of the paper, we aim to achieve an $+O(\lgn)$-approximate APSP for some specific vertex  pairs.  In the overview, we present an algorithm that achieves an $+O(\lgn)$ approximation of the length of $st$ path.  Unfortunately, the algorithm does not work in general --- it relies on a specific assumption (in addition to \Cref{eq:1}), which we describe next.

\begin{assumption} \label{ass}
	Consider an $st$ path. For each $i$ (where $0 \le i \le \lgn-2$), $\ball_{i+1}(\ui)$ contains $\uii$ and $\ball_{i+1}(\vi)$ contains $\vii$.
\end{assumption}

By definition, $\uii \in A_{i+1}$, and $\ball_{i+1}(\ui)$ cannot contain a vertex from $A_{i+1}$. Thus, this assumption is clearly incorrect.  In fact, $\ui$ may be far from $\uii$ along the $st$ path, and $\ball_{i+1}(\ui)$ may not even contain a vertex close to $\uii$. However, let us examine what this assumption leads to.  Using this assumption and \Cref{eq:1}, we prove the following lemma:

\begin{lemma}\label{lem:overviewmain}
For each $0 \le i\le \lln$, $est(\ui,\vi) \le |\ai\bi| + 4(\llni)$
\end{lemma}

The  above lemma implies an additive $O(\log \log n)$-approximation of the $st$ path. Setting $i=0$ in the lemma gives $est(u_0,v_0) \leq |a_0b_0| + 4 \log\log n$. Since $a_0 = u_0 = s$ and $b_0 = v_0 = t$, we obtain $est(s, t) \leq |st| + 4\lgn$. This yields an $+O(\lgn)$ approximate APSP and a 2-approximation for the $st$ path when $|st| \geq 4\lgn$.

We now prove the above lemma using induction on $i$, where $i$ ranges from $\lln$ to 0. For the base case $i=\lln$, consider vertices $u_{\lln}$ and $v_{\lln}$. In \Cref{sec:base}, we present a simple algorithm to approximate their shortest path without requiring \Cref{ass}. Thus, we have:

\[
est(u_{\lln}, v_{\lln}) \leq |a_{\lln} b_{\lln}| + 4
\]

This completes our base case. Assume, by induction, that we have approximated the distance between $\uii$ and $\vii$ as follows:

\begin{equation}\label{eq:2}
est(\uii, \vii) \leq |\aii \bii| + 4(\lgn - (i+1))
\end{equation}

We will now estimate the distance between $\ui$ and $\vi$. This can be done in two steps:

\begin{enumerate}[noitemsep]
	\item First, estimate the distance between $\uii$ and $\vi$
	\item Then, use the above result to improve the estimate between $\ui$ and $\vi$
\end{enumerate}
We now go over these two steps in detail.

\subsection{Estimate the distance between $\uii$ and $\vi$}

Consider the following simple algorithm to approximate the distance between $\uii$ and $\vi$.
We construct a graph $H_{\uii}$ containing:
\begin{enumerate}[noitemsep]
    \item Edges from $\uii$ to every vertex $x \in G$, weighted by $est(\uii, x)$.
    \item All edges with degree $\TL(\ti{i+1})$.
    \item For each vertex $x \in G$ and for each $i \in [0 \dots \lln]$, an edge  $[x,pivot_i(x)]$ with weight $|x~pivot_i(x)|$.\\
\end{enumerate}

The reader can verify that the size of $H_{\uii}$ is $\TL(n\ti{i+1})$. We compute the shortest paths from $\uii$ to all other vertices in $H_{\uii}$ in $\TL(n \ti{i+1})$ time. Repeating this for all vertices in $A_{i+1}$ gives a total runtime of $\TL(n^2)$. We claim that a path from $\uii$ to $\vi$ exists in $H_{\uii}$.

\[
[\uii, \vii] \odot (\vii, \bii) \odot \bii\bi \odot (\bi, \vi).
\]

\begin{center}

\begin{tikzpicture}
    \node[draw, circle, fill=black, label=below:$s$, minimum size=4pt, inner sep=0pt] (s) at (-2,0) {};
    \node[draw, circle, fill=black, label=below:$t$, minimum size=4pt, inner sep=0pt] (t) at (9,0) {};

    \draw[thick] (s) -- (t);

    \node[draw, circle, fill=black, label=below:\scalebox{0.8}{$\aii$}, minimum size=4pt, inner sep=0pt] (ui) at (0.5,0) {};
    \node[draw, circle, fill=black, label=left:\scalebox{0.8}{$\uii$}, minimum size=4pt, inner sep=0pt] (piu) at (0.8,1) {};
    \draw[thick] (ui) -- (piu);

    \node[draw, circle, fill=black, label=above:\scalebox{0.8}{$\vi$}, minimum size=4pt, inner sep=0pt] (vi) at (8, 0.5) {};
    \node[draw, circle, fill=black, label=below:\scalebox{0.8}{$\bi$}, minimum size=4pt, inner sep=0pt] (bi) at (8, 0) {};
    \draw[thick] (vi) -- (bi);

    \node[draw, circle, fill=black, label=above:\scalebox{0.8}{$\vii$}, minimum size=4pt, inner sep=0pt] (piv) at (6, 0.5) {};
	    \node[draw, circle, fill=black, label=below:\scalebox{0.8}{$\bii$}, minimum size=4pt, inner sep=0pt] (bii) at (5.5,0) {};
    \draw[thick] (piv) -- (bii);

	    \draw[ thick] (piu) to[bend left=10] node[midway, above, sloped] {\scalebox{0.7}{$est(\uii,\vii$)}} (piv);
	    \draw[line width=1mm,pink,opacity=0.6] (piu) to[bend left=10]  (piv);
	\draw[line width=1mm,pink,opacity=0.6] (piv)--(bii);

	\draw[line width=1mm,pink,opacity=0.6] (bii)--(bi);
	\draw[line width=1mm,pink,opacity=0.6] (bi)--(vi);

\end{tikzpicture}

\end{center}

Note that the first edge is a weighted edge and is added in point (1) in the above enumeration. Using \Cref{lem:onpath}, all edges in path $\bii t$ have degree $\TL(\ti{i+1})$. Thus, the subpath $\bii\bi$ is in $H_{\uii}$. The second and the fourth edge are added due to point (3) in the above enumeration. Finding the shortest path from $\uii$ in $H_{\uii}$ yields a new estimate for the distance between $\uii$ and $\vi$:

\begin{align}
    est(\uii, \vi) &\leq |[\uii, \vii] \odot (\vii, \bii) \odot \bii\bi \odot (\bi, \vi)| \notag\\
    &= |[\uii, \vii]| + |(\vii, \bii)| + |\bii\bi| + |(\bi, \vi)| \notag\\
    &= est(\uii, \vii) + 1 + |\bii\bi| + 1 \notag\\
    \intertext{Using \Cref{eq:2}, we get }
    &\leq \left(|\aii \bii| + 4(\lgn - (i+1))\right) +  |\bii \bi| + 2 \notag\\
    &= |\aii \bii| + |\bii \bi| + 4(\llni)-2 \notag\\
    &= |\aii \bi| + 4(\llni)-2 \label{eq:3}
\end{align}
\subsection{Estimate the distance between $\ui    $ and $\vi$}
Now, we use $\uii$ to estimate the distance between $\ui$ and $\vi$.
By \Cref{ass}, $\uii$ lies in $\ball_{i+1}(\ui)$.
Since $\uii$ is unknown a priori, we consider all vertices in $\ball_{i+1}(\ui)$ to estimate the distance between $\ui$ and $\vi$. The following algorithmic procedure  computes this estimate:

\begin{algorithm}[H]
\caption{Find the estimate between $\ui$ and $\vi$}
\label{alg:idea}
\ForEach{$(x,y) \in A_i \times A_i$}
{
	\ForEach{$w \in \ball_{i+1}(x)$}
	{
		$est(x,y) \le \min\{est(x,w)+est(w,y), est(x,y)\}$
	}
}
\end{algorithm}

After the end of the above algorithm, the new estimate between $\ui$ and $\vi$ will be calculated as:
\begin{align*}
    est(\ui, \vi) &\leq est(\ui,\uii) +  est(\uii,\vi)\\
    \intertext{Using \Cref{eq:3}, $est(\uii,\vi) \le |\aii \bi| + 4(\llni)-2$. Also, since $\uii \in \ball_{i+1}(\ui)$, we have calculated the exact distance from $\ui$ to $\uii$ using \Cref{lem:pivotandball}. Thus, we get:}
    &\le  |\ui\uii| + (|\aii\bi|+4(\llni)-2)\\
    \intertext{Using triangle inequality, we bound $|\ui\uii|$. Thus, we get:}
    &\le  |(\ui,\ai) \odot \ai\aii \odot (\aii,\uii)| + |\aii\bi|+4(\llni)-2\\
    &= 1+|\ai\aii|+1 + |\aii \bi| + 4(\llni)-2 \\
    &\le |\ai\bi|+4(\llni)\\
\end{align*}
This completes the proof of \Cref{lem:overviewmain}. Let us now look at the running time of the above algorithm. Using \Cref{lem:ballsize}, since the size of $\ball_{i+1}(\cdot)$ is $\TL(\ti{i+1})$, the algorithm's running time is:
\[
\TL\left(\sum_{x,y \in A_i \times A_i} \TL(\ti{i+1})\right) = \TL\left(\frac{n}{\ti{i}} \frac{n}{\ti{i}} \ti{i+1}\right) = \TL(n^2).
\]

Thus, our algorithm runs in $\TL(n^2)$ time and gives an $+O(\log\log n)$ approximation of the length of $st$ path. However, the above strategy depends on  \Cref{ass}. In reality, we do not require such a strong assumption. Even if $\ball_{i+1}(\ui)$ contains a vertex close to $\uii$, say within distance 7, we can still  obtain an $O(\lgn)$ additive approximation. But can we ensure even this weaker assumption. The key question is: can we guarantee the following closeness condition?

\begin{equation} \label{eq:4}
		\text{(Ensure Closeness)  $\ball_{i+1}(\ui)$ contains a vertex close to $\uii$, say at a distance $\leq 7$ from $\uii$}
\end{equation}

Surprisingly, the answer to this question is  yes.  In the next section, we design an algorithm $\es$ that enforces this condition --- one of our paper's main contributions. The algorithm $\es$ satisfies the following property:

\begin{lemma} \label{lem:near}
	After the execution of our algorithm $\textsc{EnsureCloseness}$, either $est(s,t) \le  2|st|$ or

	for all $0 \le i \le \lgn-1$, either
	\begin{enumerate}[label=(A\arabic*),nolistsep]
		\item\label{item:a1} $|s\ai| \le |t\bi|$ and $|s\ai|-|s~\pis| \le 3 $, or
		\item $|s\ai| \ge |t\bi|$ and $|t\bi|-|t~\pit| \le 3$
	\end{enumerate}
	Moreover, the running time of $\es$ is $\TL(n^2)$.
\end{lemma}

First, note that $\ui$ is also a candidate for the $i$-th pivot of $s$. Thus,  $|s~\pis| \le |s\ui| \le |s\ai| + |(\ai,\ui)| = |s\ai| + 1$. Thus, $|s~\pis|$ may be greater than $|s\ai|$ by at most 1.  However, $|s~\pis|$ may be much smaller than $|s\ai|$. The lemma ensures this does not happen. We now provide a pictorial interpretation of this lemma. Let $p$ be the vertex on the $st$ path at a distance $|s~\pis|$ from $s$. Similarly, let $q$ be the vertex on the $st$ path at a distance $|t~\pit|$ from $t$. The lemma states that after executing our algorithm \textsc{EnsureCloseness}, either $est(s,t) \le 2|st|$, or one of the following holds:

\begin{enumerate}
    \item the distance between $p$ and $\ai$ is $\le 3$, or
    \item the distance between $q$ and $\bi$ is $\le 3$.
\end{enumerate}

\begin{figure}[hpt!]
\centering
\begin{tikzpicture}
    \node[draw, circle, fill=black, label=below:$s$, minimum size=4pt, inner sep=0pt] (s) at (0,0) {};
    \node[draw, circle, fill=black, label=below:$t$, minimum size=4pt, inner sep=0pt] (t) at (10,0) {};
    \node[draw, circle, fill=blue, label=above:$\ai$, minimum size=4pt, inner sep=0pt] (ai) at (3.5,0) {};
    \node[draw, circle, fill=blue, label=above:$\bi$, minimum size=4pt, inner sep=0pt] (bi) at (6.5,0) {};
    \node[draw, circle, fill=gray, label=right:$\pis$, minimum size=4pt, inner sep=0pt] (pis) at (1.1,-1.5) {};
    \node[draw, circle, fill=gray, label=left:$\pit$, minimum size=4pt, inner sep=0pt] (pit) at (8.9,-1.5) {};

    \draw[thick] (s) -- (t);

    \draw[thick, dashed] (s) circle[radius=1.85];
    \draw[thick, dashed] (t) circle[radius=1.85];

    \node[draw, circle, fill=red, label=above left:$p$, minimum size=4pt, inner sep=0pt] (p) at (1.85,0) {};
    \node[draw, circle, fill=red, label=above right:$q$, minimum size=4pt, inner sep=0pt] (q) at (8.15,0) {};

    \draw[dashed] (pis) -- (s);
    \draw[dashed] (pit) -- (t);

    \draw [decorate, decoration={brace, amplitude=11pt, mirror}, thick] (p) -- (ai)
        node[midway, below=10pt] {$\le 3$};
    \draw [decorate, decoration={brace, amplitude=11pt, mirror}, thick] (bi) -- (q)
        node[midway, below=10pt] {$\le 3$};

\end{tikzpicture}
\caption{$\es$ ensures that $(p,\ai)$ or $(q,\bi)$ are close to each other.}
\label{fig:2}
\end{figure}

The reader may ask how does \Cref{lem:near} help us remove  \Cref{ass} or ensure the weaker closeness condition in (\ref{eq:4})? Let us assume that after executing $\es$, $est(s,t) > |st|$. Also, let us assume $|s\aii| \le |t\bii|$ -- otherwise, a symmetric argument applies from the $t$ side. Consider the vertex $\ui$.  We will now show that $\ball_{i+1}(\ui)$ contains a vertex close to $\uii$. There are two cases:

\begin{enumerate}
\item $|\ai \aii| \le 5 $

By triangle inequality, $|\ui\uii| \le |(\ui,\ai)| + |\ai\aii| + |(\ai,\ui)|$. If $|\ai \aii| \le 5$, then $|\ui \uii| \le 7$, implying $\ui$ and $\uii$ are close. Thus, (\ref{eq:4}) holds trivially in this case.

\item $|\ai \aii| \ge 6$

Let $p$ be the vertex at a distance $|s~\piis|$ from $s$ on the $st$ path. By \Cref{lem:near}, $p$ lies within distance 3 of $\aii$ on the $st$ path.  
Let $z$ be the vertex exactly 3 away from $p$ on the $sp$ path, or equivalently, within distance 6 of $\aii$ on $s\aii$ (see \Cref{fig:containsy}). Ideally, we want $\uii \in \ball_{i+1}(\ui)$, but this is not  possible. Instead, we prove $z \in \ball_{i+1}(\ui)$, ensuring {\em closeness}. To this end, we show the following important lemma:

\begin{figure}[hpt!]
\centering
\begin{tikzpicture}
    \node[draw, circle, fill=black, label=below:$s$, minimum size=4pt, inner sep=0pt] (s) at (-2,0) {};
    \node[draw, circle, fill=black, label=below:$t$, minimum size=4pt, inner sep=0pt] (t) at (9,0) {};
    \node[draw, circle, fill=black, label=below:\scalebox{0.8}{$\aii$}, minimum size=4pt, inner sep=0pt] (ai) at (3.5,0) {};

    \draw[thick] (s) -- (t);

    \node[draw, circle, fill=red, label=below:\scalebox{0.8}{$p$}, minimum size=4pt, inner sep=0pt] (p) at (2.3,0) {};

    \node[draw, circle, fill=black, label=above:\scalebox{0.8}{$\ui$}, minimum size=4pt, inner sep=0pt] (ui1) at (1, 0.5) {};
    \node[draw, circle, fill=black, label=below:\scalebox{0.8}{$\ai$}, minimum size=4pt, inner sep=0pt] (ai1) at (0.8, 0) {};
    \draw[thick] (ui1) -- (ai1);

    \draw[thick, dotted] (ui1) circle[radius=0.7];

    \node at (-0.9, 0.9) (label) {\scalebox{0.8}{$\ball_{i+1}(\ui)$}};
    \draw[->] (label.east) -- (0.3, 0.9);

    \node[draw, circle, fill=black, label=above:\scalebox{0.8}{$\uii$}, minimum size=4pt, inner sep=0pt] (ui) at (3.7, 0.5) {};

    \draw[thick] (ui) -- (ai);
	\draw [decorate, decoration={brace, amplitude=6pt, mirror}, thick] (2.3,-0.5) -- (3.5,-0.5)
        node[midway, below=6pt] {\scalebox{0.8}{$\le 3$}};
    \draw [decorate, decoration={brace, amplitude=4pt, mirror}, thick] (1.5,-0.5) -- (2.25,-0.5)
        node[midway, below=4pt] {\scalebox{0.8}{$3$}};
    \node[draw, circle, fill=blue, label=below:\scalebox{0.8}{$z$}, minimum size=4pt, inner sep=0pt] (y) at (1.5, 0) {};

\end{tikzpicture}

\caption{$p$ is the vertex on the $st$ path such that $|sp| = |s~\piis|$. $\ball_{i+1}(\ui)$ contains $z
$ which is close to $p$ and $p$ is close to $\uii$.}
\label{fig:containsy}
\end{figure}
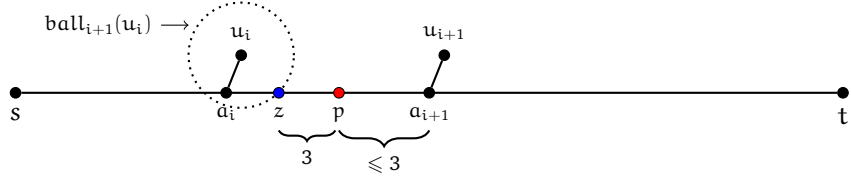

\begin{lemma}\label{lem:es}
	Assume \ref{item:a1} is true after the execution of $\es$ and $|\ai\aii| \ge 6$ for some $i \in [0 \dots \log\log n-2]$. Let $p$ be the vertex on the $st$ path at a distance $|s~\piis|$ from $s$.
	 Let $z$ be the vertex at a distance exactly 3 from $p$ on $sp$ path or at a distance $\le  6$ from $\aii$ on $s\aii$ path (See \Cref{fig:containsy}). Then, $z \in \ball_{i+1}(\ui)$.
\end{lemma}

\begin{proof}
Suppose, for contradiction, that $\ball_{i+1}(\ui)$ does not contain $z$. Then, the pivot selection ensures
\begin{equation} \label{eq:5}
	|\ui pivot_{i+1}(\ui)| \le  |\ui z| \le |(\ui, \ai) \odot \ai z| = 1 + |\ai z|.
\end{equation}

This implies the existence of a path from $s$ to $pivot_{i+1}(\ui)$ of length:

\begin{align*}
     |s \ai \odot (\ai, \ui) \odot \ui pivot_{i+1}(\ui)| &= |s \ai|+|(\ai, \ui)| + |\ui pivot_{i+1}(\ui)|\\
     \intertext{Using \Cref{eq:5}, $|\ui pivot_{i+1}(\ui)| \le 1+|\ai z|$. Thus, we get:}
     &\le |s \ai| + 1 + 1 + |\ai z| \\
    &= |s z| + 2 \\
 	\intertext{But $z$ is at a distance 3 away from $p$, Thus, $|sz| = |sp|-3$. Thus, we get}
    &= |s p|-3+2\\
    &=|sp|-1
    \intertext{Using lemma's condition, $|sp| = |s~pivot_{i+1}(s)|$. Thus, we get:}
    &= |s~pivot_{i+1}(s)|-1
\end{align*}

Thus, $pivot_{i+1}(\ui)$ lies at a distance  $<|s~pivot_{i+1}(s)|$ from $s$.  But $pivot_{i+1}(\ui) \in A_{i+1}$ and by definition, no vertex in $A_{i+1}$ can be closer to $s$ than $pivot_{i+1}(s)$, leading to a contradiction. Hence, $\ball_{i+1}(\ui)$ must contain $z$.

\end{proof}

\end{enumerate}

 This concludes the overview of our algorithm.

\section{Organisation of the paper}
While working on this problem, we initially proved \Cref{thm:main} for $k = \log n$ and later realized that it can be generalized to any $k$. To reflect this progression, we  present our results in the order we discovered it. 
In \Cref{sec:es}, we prove \Cref{lem:near}. Armed with this tool, we prove \Cref{thm:main} when $k = \log n$ in \Cref{sec:main} and \Cref{sec:general}. In \Cref{sec:genk}, we generalize our result and prove \Cref{thm:main}.

\section{Ensuring  closeness}
\label{sec:es}

In this section, we will prove \Cref{lem:near}, which will ensure the closeness condition in (\ref{eq:4}). To this end, we will crucially use the  following lemma which follows directly from the fact that we have precomputed a $(2,1)$-approximation of the shortest path.

\begin{lemma} \label{lem:use21}
    For each $i \in [0,\lln]$, $est(\ui, \vi) \leq 2|\ai \bi| + 5$.
\end{lemma}

\begin{proof}
Since we have precomputed $(2,1)$-approximation of the length of the path between $\ui$ and $\vi$, we get:
\begin{align*}
    est(\ui, \vi) &\leq 2|\ui \vi| + 1 \\
    &\leq 2(|(\ui, \ai)| + |\ai \bi| + |(\bi, \vi)|) + 1\hspace{0.8cm}(\text{Using triangle inequality})\\
    &= 2(1 + |\ai \bi| + 1) + 1\\
    &= 2|\ai \bi| + 5
\end{align*}
\end{proof}

Let us now describe $\es(i)$. We call this function for all $i \in [0 \dots \lln]$. The algorithm can be divided into three parts. We now describe each part in detail.

\subsection{Part 1: Estimate the distance between $\vi$ and $\pis$}

In Part 1 of $\es(i)$, we process all vertices in $A_i$, including $\vi$. In $\uf(\vi,i)$, we construct a graph $H_{\vi}$ containing
	\begin{enumerate}[noitemsep]
    \item Edges from $\vi$ to every vertex $x \in G$, weighted by $est(\vi, x)$.
    \item All edges with degree $\TL(\ti{i})$.
    \item For each vertex $x \in G$ and each $i \in [0\dots \lln]$, an edge  $[x,pivot_i(x)]$ with weight $|x~pivot_i(x)|$.
\end{enumerate}

At most $\TL(n)$ edges are added in steps (1) and (3), while step (2) adds at most $\TL(n \ti{i})$ edges. Thus, the size of $H_{\vi}$ is $\TL(n \ti{i})$. We then compute shortest paths from $\vi$ in $H_{\vi}$, updating $est(\vi, \cdot)$ as needed. The total runtime for Part 1 of $\es(i)$, summed over all vertices in $A_i$, is
\[
\sum_{w \in A_i} \TL(n \ti{i}) = \TL\left(\frac{n}{\ti{i}} \times n \ti{i} \right) = \TL(n^2).
\]

\begin{minipage}{0.335\textwidth}
    \begin{algorithm}[H]
        \caption{\textsc{$\es(i)$}}
		{\color{blue}\tcc{Part 1}}
		\ForEach{$w \in A_i$}
		{
			$\uf(w,i)$
		}
		{\color{blue}\tcc{Part 2}}
		\ForEach{$w \in A_i$}
		{
			$\uf(w,i)$
		}
		{\color{blue}\tcc{Part 3}}
		\ForEach{vertex $s \in G$}
		{
			$\tri(s,i)$
		}
    \end{algorithm}
\end{minipage}
\begin{minipage}{0.62\textwidth}
    \begin{algorithm}[H]
        \caption{$\uf(w,i)$}
		Make a graph $H_w$ that contains:\\
		-- An edge from $w$ to all $x \in G$ with weight $est(w,x)$ \label{line:uf1}\\

		-- All edges with degree $\TL(\ti{i})$\label{line:uf2}\\
		-- For each vertex $x \in G$ and for each $i \in [0 \dots \lln]$, an edge  $[x,pivot_i(x)]$ with weight $|x~pivot_i(x)|$\label{line:uf3}\\
		Find the shortest path from $w$ to all other vertices in $H_w$ and update $est(w,\cdot)$ and $est(\cdot,w)$    \end{algorithm}
    \begin{algorithm}[H]
        \caption{$\tri(s,i)$}
		\ForEach{ $t \in G$}
		{
			$est(s,t) \leftarrow \min\{est(s,t), |s~\pis| + est(\pis,t) \}$\\
		}
    \end{algorithm}
\end{minipage}

We claim that Part 1 of the algorithm would give us a new estimate between $\vi$ and $\pis$. Indeed, there is a path from $\vi$ to $\pis$ in $H_{\vi}$ as follows:
	$$[\vi,\ui] \odot (\ui,\ai) \odot \ai s \odot [s,\pis]$$
	The first edge is added in \Cref{line:uf1} of $\uf(\vi,i)$. The second and last edge are added in \Cref{line:uf3}. Also, using \Cref{lem:onpath}, each edge in $s\ai$ has degree $\TL(\ti{i})$ and is added in \Cref{line:uf2}. Thus, the above path is in $H_{\vi}$ and $est(\vi,\pis)$ must be updated to:
\begin{center}
\begin{tikzpicture}
    \node[draw, circle, fill=black, label=below:$s$, minimum size=4pt, inner sep=0pt] (s) at (-2,0) {};
    \node[draw, circle, fill=black, label=below:$t$, minimum size=4pt, inner sep=0pt] (t) at (9,0) {};

    \draw[thick] (s) -- (t);


    \node[draw, circle, fill=black, label=left:\scalebox{0.8}{$\pis$}, minimum size=4pt, inner sep=0pt] (pis) at (-1,1) {};
    \draw[thick] (s) -- (pis);

    \node[draw, circle, fill=black, label=above:\scalebox{0.8}{$\ui$}, minimum size=4pt, inner sep=0pt] (ui) at (2, 0.5) {};
    \node[draw, circle, fill=black, label=below:\scalebox{0.8}{$\ai$}, minimum size=4pt, inner sep=0pt] (ai) at (2, 0) {};
    \draw[thick] (ui) -- (ai);

    \node[draw, circle, fill=black, label=above:\scalebox{0.8}{$\vi$}, minimum size=4pt, inner sep=0pt] (vi) at (7, 0.5) {};
	    \node[draw, circle, fill=black, label=below:\scalebox{0.8}{$\bi$}, minimum size=4pt, inner sep=0pt] (bi) at (6.5,0) {};
    \draw[thick] (vi) -- (bi);

	    \draw[ thick] (ui) to[bend left=10] node[midway, above] {\scalebox{0.7}{$est(\vi, \ui)$}} (vi);
	    \draw[line width=1mm,pink,opacity=0.6] (ui) to[bend left=10]  (vi);
	    \draw[line width=1mm,pink,opacity=0.6] (ui)--(ai);
	    \draw[line width=1mm,pink,opacity=0.6] (ai)--(s);
	    \draw[line width=1mm,pink,opacity=0.6] (s)--(pis);

\end{tikzpicture}
\end{center}
\begin{align}
    est(\vi,\pis) & \leq |[\vi,\ui] \odot (\ui, \ai) \odot  \ai s \odot [s,\pis]| \notag\\
    &\leq est(\vi,\ui) + |(\ui, \ai)| + |\ai s| + |[s, \pis]| \notag \\
    \intertext{Using \Cref{lem:use21}, $est(\vi,\ui) \le 2|\ai\bi|+5$. Thus, we get:}
    &\le  (2|\ai\bi|+5) +1+|s\ai|+|s~\pis| \notag \\
    &= |s~\pis| + |s\ai|+2|\ai\bi|+6 \label{eq:6}
\end{align}

By symmetry, we also get a better estimate between $\ui$ and $\pit$.

\subsection{Part 2: Estimate the distance from $\pis$ to $t$}
The second part of $\es(i)$ is same as the first one. So, its running time is $\TL(n^2)$. Consider the function $\uf(\pis,i)$. We claim that there is a path from $\pis$ to $t$ in $H_{\pis}$ as follows:
	$$[\pis,\vi] \odot (\vi,\bi) \odot \bi t$$.
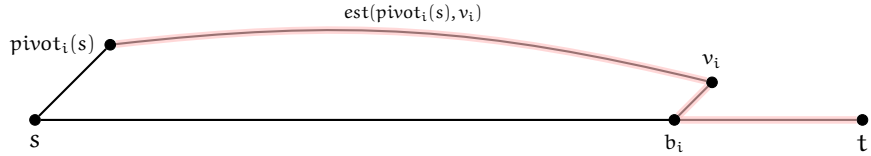
\begin{center}

\begin{tikzpicture}
    \node[draw, circle, fill=black, label=below:$s$, minimum size=4pt, inner sep=0pt] (s) at (-2,0) {};
    \node[draw, circle, fill=black, label=below:$t$, minimum size=4pt, inner sep=0pt] (t) at (9,0) {};

    \draw[thick] (s) -- (t);


    \node[draw, circle, fill=black, label=left:\scalebox{0.8}{$\pis$}, minimum size=4pt, inner sep=0pt] (pis) at (-1,1) {};
    \draw[thick] (s) -- (pis);

    \node[draw, circle, fill=black, label=above:\scalebox{0.8}{$\vi$}, minimum size=4pt, inner sep=0pt] (vi) at (7, 0.5) {};
	    \node[draw, circle, fill=black, label=below:\scalebox{0.8}{$\bi$}, minimum size=4pt, inner sep=0pt] (bi) at (6.5,0) {};
    \draw[thick] (vi) -- (bi);

	    \draw[ thick] (pis) to[bend left=10] node[midway, above] {\scalebox{0.7}{$est(\pis, \vi)$}} (vi);
	    \draw[line width=1mm,pink,opacity=0.6] (pis) to[bend left=10]  (vi);
	    \draw[line width=1mm,pink,opacity=0.6] (vi)--(bi);
	    \draw[line width=1mm,pink,opacity=0.6] (bi)--(t);

\end{tikzpicture}

\end{center}

	The first edge is added in \Cref{line:uf1} of $\uf(\pis,i)$. The second edge $(\vi,\bi)$ is added in \Cref{line:uf3}. Also, using \Cref{lem:onpath}, each edge in $\bi t$ has degree $\TL(\ti{i})$ and is added in \Cref{line:uf2}. Thus, the above path is in $H_{\pis}$ and $est(\pis,t)$ must be updated to:
\begin{align}
	est(\pis,t)&\le |[\pis,\vi] \odot (\vi,\bi) \odot \bi t|\notag\\
	&= est(\pis,\vi) + |(\vi, \bi)| + |\bi t| \notag\\
	\intertext{Using \Cref{eq:6}, we get}
	&\le\left(|s~\pis| + |s\ai|+2|\ai\bi|+6\right) + 1+ |\bi t|\notag\\
	&=|s~\pis| + |s\ai|+2|\ai\bi|+ |\bi t|+7\label{eq:7}
\end{align}

By symmetry, we also get a better estimate between $\pit$ and $s$.
\subsection{Part 3: Estimate the distance between $s$ and $t$}

In Part 3 of $\es(i)$, we call the procedure $\tri(s,i)$. $\tri(s,i)$ estimates the shortest path from $s$ to each $t \in G$ using  $\pis$.  For a fixed $s$, the running time of $\tri(s,i)$ is $O(n)$. Summing over all vertices, the running time of Part 3 of our algorithm is $O(n^2)$. After executing $\tri(s,i)$, we get a better estimate between $s$ and $t$ as follows:
\begin{align}
    est(s,t) &\leq |s~\pis| +est(\pis,t) \notag\\
    \intertext{Using \Cref{eq:7}, we get}
    &=|s~\pis| + (|s~\pis| + |s\ai|+2|\ai\bi|+ |\bi t|+7) \notag \\
    &=2|s~\pis| +  |s\ai|+2|\ai\bi|+ |\bi t|+7 \label{eq:8}
\end{align}
Again, by symmetry, we can show that $est(s,t) \le 2|t~\pit| + |t\bi|+2|\ai\bi| +|s\ai|+7$ after Part 3 of our algorithm.
We are now ready to prove \Cref{lem:near}.

\subsection{Proof of \Cref{lem:near}}

For each $i$, $\es(i)$ takes $\TL(n^2)$ time. Thus, the total time taken by $\es$ for all $i$ is also $\TL(n^2)$.

Let us assume that $est(s,t) > 2|st|$ after the execution of $\es(i)$. Also, let us assume that  $|s\ai| \le |t\bi|$ --- a symmetrical argument applies when $|t\bi| \le |s\ai|$. Let us assume for contradiction that  $|s\ai|-|s~\pis| \ge  4$, or $|s~\pis| \le |s\ai|-4$. Putting this in \Cref{eq:8}, we get:
\begin{align}
    est(s,t) &\leq 2|s~\pis| +  |s\ai|+2|\ai\bi|+ |\bi t|+7\notag \\
    &\leq 2(|s\ai| - 4) + |s\ai| + 2|\ai\bi| + |\bi t| + 7 \notag \\
    &\leq 2|s\ai| + |s\ai| + 2|\ai\bi| + |\bi t| -1 \notag \\
    \intertext{Since $|s\ai|\le |t\bi|$, $2|s\ai| \le |s\ai|+|t\bi|$,}
    &\leq (|s\ai| +|t\bi|)  + |s\ai| + 2|\ai\bi| + |\bi t| -1\notag \\
    &= 2(|s\ai| + |\ai\bi| + |\bi t|) - 1 \notag \\
    &= 2|st| - 1\notag
\end{align}
But this is a contradiction as we assumed that $est(s,t) > 2|st|$ after executing $\es(i)$. This implies that our assumption is wrong. Thus, $|s\ai|-|s~\pis| \le  3$. This completes the proof of \Cref{lem:near}.

\section{Our Main Algorithm}
\label{sec:main}
In the next two sections, we will focus on $st$ pair for which $est(s,t)> 2|st|$ after executing $\es$. All our lemmas will have this as its pre-condition and for brevity, we will not mention it in our lemmas.
Our algorithm will run in $\lgn$ iteration from $i=\lln$ to $0$. We will show the following main lemma.

\begin{lemma} \label{lem:mainclaim}
	 After the  $i$-th iteration of our algorithm, $est(\ui,\vi) \le |\ai\bi|+18(\lgn-i)$.
\end{lemma}

Notice that the above algorithm directly implies  \Cref{thm:main} when $k= \log n$. Indeed, after  the last iteration of our algorithm, that is when $i=0$,
$$est(u_0,v_0) \le  |a_0b_0| + 18\lgn$$
Using the fact that $u_0 =a_0 = s$ and $v_0 =b_0=t$, we get: $est(s, t) \le |s t| + 18\lgn$. This yields a 2-approximation of $st$ path when $|st| \geq 18\lgn$. This proves \Cref{thm:main} when $k= \log n$. The next two sections focus on proving the above lemma using induction on $i$ where $i$ runs from $\lln$ to 0. Let us see the base case first.

\subsection{Our Algorithm: The base case}
\label{sec:base}
In this section, we design an algorithm that will give us a good estimate between  $u_{\lgn-1}$ and $v_{\lgn-1}$. Using \Cref{eq:1}, let us assume that  the highest-degree vertex on the $st$ path has a degree in the range $[2^\ell,2^{\ell+1}]$, where $\ell$ is an integer in $[\log n/2 \dots  \log n]$ . Since $\ell$ is unknown in advance, we process all $\ell \in [\log n/2 \dots  \log n]$. For each $\ell$, we sample a set $B_\ell$ from $G$ with probability $\frac{1}{2^\ell}$. With high probability, the size of $B_\ell$ is $\TL(n/2^\ell)$. We now use the following algorithm to compute the shortest path from a vertex in $B_\ell$.

\begin{algorithm}
\caption{The base case}
\label{alg:base}
\ForEach{$\ell \in [\log n/2, \log n]$}
{
\ForEach{$w$ in $B_\ell$}
	{
		Make a new graph $H_w$ that contains:\\
		-- All edges adjacent to $w$ in $G$\label{line:base4}\\
		-- All edges of degree $\le \TL(2^{\ell+1})$ \label{line:base5}\\
		Find the shortest path from $w$ in $H_w$ and update $est(w,\cdot)$ and $est(\cdot,w)$ \\

		\ForEach{$x,y \in A_{\lgn-1} \times A_{\lgn-1}$}
		{
			$est(x,y) \leftarrow \min\{ est(x,y), est(x,w) + est(w,y)\}$\\
		}
	}
}



\end{algorithm}

We iterate over all vertices in $B_\ell$ for each $\ell \in [\log n/2, \log n]$. Fix a $\ell$ and consider a vertex $w \in B_\ell$. We construct $H_w$ that contains all edges adjacent to $w$ in $G$ and all edges of degree $\TL(2^{\ell+1})$. Thus, the size of $H_w$ is $\TL(n 2^{\ell+1})$. We then find shortest path from $w$  in $H_w$ and update $est(w,\cdot)$ as required.

Next, we update our estimate between all pairs in $A_{\lgn-1}$. Since, with a high probability, $A_{\lgn-1}$ contains $\TL(\frac{n}{2^{2^{\lln}}}) = \TL(\sqrt{n})$ vertices, this step takes $\TL(n)$ time. Thus, for a single $w \in B_\ell$, the total time taken is $\TL(n 2^{\ell+1})$. Summing over all vertices in $B_\ell$ and all $\ell$, the total running time is:
\begin{align*}
    \sum_{\ell=\log n/2}^{\log n} \sum_{w \in B_\ell} \TL(n2^{\ell+1})
    &= \sum_{\ell=\log n/2}^{\log n} \TL\left(\frac{n}{2^\ell} \times n 2^{\ell+1}\right) \\
    &= \sum_{\ell=\log n/2}^{\log n} \TL(n^2) \\
    &= \TL(n^2).
\end{align*}
Using \Cref{eq:1}, let $p$ be the vertex on $st$ path with max degree, say in the range  $[2^\ell,2^{\ell+1})$ where $\ell > \log n/2$.  Using \Cref{lem:onpath}, $p$ cannot be on $sa_{\lln}$ or $tb_{\lln}$ path. Thus, it lies between $a_{\lln}$ and $b_{\lln}$ on the $st$ path (or $a_{\lln}$ or $b_{\lln}$ may be same as $p$). With a high probability,   there is a vertex, say $w \in B_\ell$ adjacent to $p$.   We claim that $H_w$ will contain a path from $w$ to $u_{\lln}$ as follows:
  $$(w,p)\odot p\ra a_{\lgn-1} \odot (a_{\lgn-1},u_{\lgn-1})$$ 
  
  
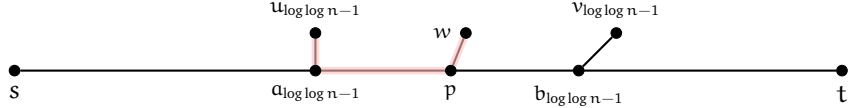
\begin{center}

\begin{tikzpicture}
    \node[draw, circle, fill=black, label=below:$s$, minimum size=4pt, inner sep=0pt] (s) at (-2,0) {};
    \node[draw, circle, fill=black, label=below:$t$, minimum size=4pt, inner sep=0pt] (t) at (9,0) {};

    \draw[thick] (s) -- (t);

    \node[draw, circle, fill=black, label=below:\scalebox{0.8}{$p$}, minimum size=4pt, inner sep=0pt] (p) at (3.8,0) {};
    \node[draw, circle, fill=black, label=left:\scalebox{0.8}{$w$}, minimum size=4pt, inner sep=0pt] (w) at (4,0.5) {};
    \draw[thick] (p) -- (w);

    \node[draw, circle, fill=black, label=above:\scalebox{0.8}{$u_{\lln}$}, minimum size=4pt, inner sep=0pt] (ui1) at (2, 0.5) {};
    \node[draw, circle, fill=black, label=below:\scalebox{0.8}{$a_{\lln}$}, minimum size=4pt, inner sep=0pt] (ai1) at (2, 0) {};
    \draw[thick] (ui1) -- (ai1);

    \node[draw, circle, fill=black, label=above:\scalebox{0.8}{$v_{\lln}$}, minimum size=4pt, inner sep=0pt] (ui) at (6, 0.5) {};
	    \node[draw, circle, fill=black, label=below:\scalebox{0.8}{$b_{\lln}$}, minimum size=4pt, inner sep=0pt] (ai) at (5.5,0) {};
    \draw[thick] (ui) -- (ai);

	    \draw[line width=1mm,pink,opacity=0.6] (ui1)--(ai1);
	    \draw[line width=1mm,pink,opacity=0.6] (ai1)--(p);
	    \draw[line width=1mm,pink,opacity=0.6] (p)--(w);

\end{tikzpicture}

\end{center}
  The first edge is due to \Cref{line:base4} of \Cref{alg:base}. The second path and the third edge is due to \Cref{line:base5} and the fact that $p$ is the vertex of highest degree in $st$ path. After we update the estimates, $est(u_{\lgn-1},w) \le |(w,p)| + |pa_{\lln}| +|(a_{\lln},u_{\lln})|=|pa_{\lgn-1}|+2$. Similarly, $est(w, v_{\lgn-1}) \le |pb_{\lgn-1}|+2$.  Since we  estimated the distance between  $u_{\lgn-1}$ and $v_{\lgn-1}$ using $w$ in the above algorithm, we have, 
 \begin{align}
est(u_{\lgn-1},v_{\lgn-1}) &\le est(u_{\lgn-1},w) + est(w,v_{\lgn-1})\notag\\
 & \le (|pa_{\lgn-1}|+2) + (|pb_{\lgn-1}|+2)\nonumber\\
	& = |a_{\lgn-1}b_{\lgn-1}|+4\notag
\end{align}

Thus, we have proven \Cref{lem:mainclaim} for the base case, i.e., when $i = \lln$. We now proceed to the general case, where $i < \lln$.

\section{Our Algorithm: The General Case}
\label{sec:general}
Our algorithm for the general case runs for $\lgn-1$ iterations, starting from $i = \lgn-2$ down to $0$. We have already established \Cref{lem:mainclaim} for the base case. Now, assuming the induction hypothesis, we suppose that \Cref{lem:mainclaim} holds after the $(i+1)$-th iteration. We will now prove that it also holds for the $i$-th iteration. 

Our general-case algorithm closely resembles \Cref{alg:idea} from the overview. Specifically, Parts 1 and 2 of \Cref{alg:main} prepare estimates used in Part 3, which is nearly identical to \Cref{alg:idea}. We now examine each part in detail.

\begin{algorithm}
\caption{The general case}
\label{alg:main}

\ForEach{$i$ from $\llnt$ to 0 }
{

	{\color{blue}\tcc{Part 1}}
	\ForEach{$w$ in $A_{i+1}$}
	{
		\uf$(w,i+1)$\\
	}

	{\color{blue}\tcc{Part 2}}
	\ForEach{$w$ in $A_{i+1}$}
	{
		\uf$(w,i+1)$\\
	}

{\color{blue}\tcc{Part 3}}
	\ForEach{$x,y$ in $A_i \times A_i$}
	{


		\ForEach{$w$ such that $w \in ball_{i+1}(x)$ or $x = pivot_{i}(w)$}
		{

				\label{line:gen8}
				$
				est(x,y) = min
				\begin{cases}
					est(x,w)+est(w,pivot_{i+1}(w)) + est(pivot_{i+1}(w),y)\\
					est(x,y)
				\end{cases} \label{line:gen9}
				$
		}

	}

}

\end{algorithm}

\subsection{Part 1: Estimate the distance from $\vii$ to some other vertices in $A_{i+1}$}
\label{sec:part1}

In {\em Part 1} of our algorithm, we process all the vertices of $A_{i+1}$.   Thus, we process $\vii$ too. We want to get a better estimate between $\vii$ and some vertices in $A_{i+1}$.   We make the following claim:
\begin{lemma}\label{lem:basepart2}
Let $q$ be a vertex closer to $s$ than $\aii$ on the $st$ path, i.e., $|sq| \leq |s\aii|$.  Then, after  executing Part 1 of \Cref{alg:main},

$$est(\vii, pivot_{i+1}(q)) \le |q~pivot_{i+1}(q)|+|q\bii| + 18(\lgn-(i+1))+1$$
\end{lemma}

\begin{proof}

In Part 1 of our algorithm, we call the procedure $\uf(\vii,i+1)$. In this procedure, we construct a graph $H_{\vii}$.
We claim that there is a path from $\vii$ to $pivot_{i+1}(q)$ in $H_{\vii}$. The path is as follows:
$$[\vii,\uii] \odot (\uii,\aii) \odot \aii \ra q \odot [q , pivot_{i+1}(q)]$$
\begin{center}

\begin{tikzpicture}
    \node[draw, circle, fill=black, label=below:$s$, minimum size=4pt, inner sep=0pt] (s) at (-2,0) {};
    \node[draw, circle, fill=black, label=below:$t$, minimum size=4pt, inner sep=0pt] (t) at (9,0) {};

    \draw[thick] (s) -- (t);

    \node[draw, circle, fill=black, label=below:\scalebox{0.8}{$q$}, minimum size=4pt, inner sep=0pt] (x) at (0.5,0) {};
    \node[draw, circle, fill=black, label=left:\scalebox{0.8}{$pivot_{i+1}(q)$}, minimum size=4pt, inner sep=0pt] (pix) at (0.8,1) {};
    \draw[thick] (x) -- (pix);

    \node[draw, circle, fill=black, label=above:\scalebox{0.8}{$\uii$}, minimum size=4pt, inner sep=0pt] (ui1) at (2, 0.5) {};
    \node[draw, circle, fill=black, label=below:\scalebox{0.8}{$\aii$}, minimum size=4pt, inner sep=0pt] (ai1) at (2, 0) {};
    \draw[thick] (ui1) -- (ai1);

    \node[draw, circle, fill=black, label=above:\scalebox{0.8}{$\vii$}, minimum size=4pt, inner sep=0pt] (ui) at (7, 0.5) {};
	    \node[draw, circle, fill=black, label=below:\scalebox{0.8}{$\bii$}, minimum size=4pt, inner sep=0pt] (ai) at (6.5,0) {};
    \draw[thick] (ui) -- (ai);

	    \draw[ thick] (ui) to[bend left=-10] node[midway, below] {\scalebox{0.7}{$est(\vii, \uii)$}} (ui1);
	    \draw[line width=1mm,pink,opacity=0.6] (ui) to[bend left=-10]  (ui1);
	    \draw[line width=1mm,pink,opacity=0.6] (ui1)--(ai1);
	    \draw[line width=1mm,pink,opacity=0.6] (ai1)--(x);
	    \draw[line width=1mm,pink,opacity=0.6] (x)--(pix);

\end{tikzpicture}

\end{center}
The first weighted edge is added in \Cref{line:uf1} of $\uf(\vii,i+1)$. The second and fourth edge are added in \Cref{line:uf3}. Using \Cref{lem:onpath}, each edge in $\aii q$ has degree $\TL(\ti{i+1})$ and is added in \Cref{line:uf2}.   Thus, we update the estimate between $\vii$ to $pivot_{i+1}(q)$ as follows:

\begin{align}
	est(\vii,pivot_{i+1}(q)) &\le |[\vii,\uii] \odot (\uii,\aii) \odot \aii \ra q \odot [q, \ra pivot_{i+1}(q)]|\nonumber\\
	&\le est(\vii,\uii) + |(\uii,\aii)| +|\aii q| + |q~pivot_{i+1}(q)| \nonumber\\
\intertext{Using induction hypothesis, $est(\vii,\uii) \le |\aii \bii| + 18(\lgn -(i+1))$. Thus, we get:}
	&\le \left(|\aii \bii|+18(\lgn-(i+1))\right) + 1+|\aii q| + |q~pivot_{i+1}(q)| \nonumber\\
	&= |q~pivot_{i+1}(q)| + |q\bii|+18(\lgn-(i+1))+1\nonumber
\end{align}

\end{proof}

By symmetry, we also get a better estimate between $\uii$ and some select vertices in $A_{i+1}$. We claim the following lemma.

\begin{lemma}
Let $q$ be a vertex closer to $t$ than $\bii$ on the $st$ path, i.e., $|tq| \le |t\bii|$.  Then, after  executing  Part 1 of \Cref{alg:main},

$$est(\uii, pivot_{i+1}(q)) \le |q~pivot_{i+1}(q)|+|q\aii| + 18(\lgn-(i+1))+1$$
\end{lemma}

\subsection{Part 2: Estimate the distance from vertices of $A_{i+1}$ and $\vi$}

In Part 2 of our algorithm, we will give a new estimate between some vertices of $A_{i+1}$  and $\vi$. Similar to  the previous section, we will prove the following statement.

\begin{lemma} \label{lem:basepart3}
Let $q$ be a vertex closer to $s$ than $\aii$ on the $st$ path, i.e., $|sq| \leq |s\aii|$.  Then, after  executing Part 2 of \Cref{alg:main},

$$est(pivot_{i+1}(q),\vi) \le |q~pivot_{i+1}(q)|+|q\bi| + 18(\lgn-(i+1))+3$$
\end{lemma}

\begin{proof}

In Part 2 of our algorithm, we call the procedure $\uf$ all vertices of $A_{i+1}$. So, we also call the procedure $\uf(pivot_{i+1}(q),i+1)$. In this procedure, we construct a graph $H_{pivot_{i+1}(q)}$. We claim that there is a path from $pivot_{i+1}(q)$ to $\vi$ in $H_{pivot_{i+1}(q)}$ as follows:

$$[pivot_{i+1}(q),\vii] \odot (\vii,\bii) \odot \bii \ra \bi \odot (\bi,\vi)$$
The first weighted edge is added in \Cref{line:uf1} of $\uf(pivot_{i+1}(q),i+1)$. The second and fourth edge are added in \Cref{line:uf3}.  Using \Cref{lem:onpath}, each edge in $\bii \bi$ has degree $\TL(\ti{i+1})$ and is added in \Cref{line:uf2}.   Thus, we update the estimate between  $pivot_{i+1}(q)$ and $\vi$ as follows:

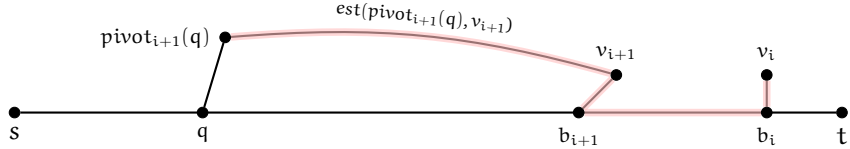
\begin{center}

\begin{tikzpicture}
    \node[draw, circle, fill=black, label=below:$s$, minimum size=4pt, inner sep=0pt] (s) at (-2,0) {};
    \node[draw, circle, fill=black, label=below:$t$, minimum size=4pt, inner sep=0pt] (t) at (9,0) {};

    \draw[thick] (s) -- (t);

    \node[draw, circle, fill=black, label=below:\scalebox{0.8}{$q$}, minimum size=4pt, inner sep=0pt] (x) at (0.5,0) {};
    \node[draw, circle, fill=black, label=left:\scalebox{0.8}{$pivot_{i+1}(q)$}, minimum size=4pt, inner sep=0pt] (pix) at (0.8,1) {};
    \draw[thick] (x) -- (pix);

    \node[draw, circle, fill=black, label=above:\scalebox{0.8}{$\vi$}, minimum size=4pt, inner sep=0pt] (ui1) at (8, 0.5) {};
    \node[draw, circle, fill=black, label=below:\scalebox{0.8}{$\bi$}, minimum size=4pt, inner sep=0pt] (ai1) at (8, 0) {};
    \draw[thick] (ui1) -- (ai1);

    \node[draw, circle, fill=black, label=above:\scalebox{0.8}{$\vii$}, minimum size=4pt, inner sep=0pt] (ui) at (6, 0.5) {};
	    \node[draw, circle, fill=black, label=below:\scalebox{0.8}{$\bii$}, minimum size=4pt, inner sep=0pt] (ai) at (5.5,0) {};
    \draw[thick] (ui) -- (ai);

	    \draw[ thick] (ui) to[bend left=-10] node[midway, above, sloped] {\scalebox{0.7}{$est(pivot_{i+1}(q),\vii$)}} (pix);
	    \draw[line width=1mm,pink,opacity=0.6] (ui) to[bend left=-10]  (pix);
	\draw[line width=1mm,pink,opacity=0.6] (ui1)--(ai1);
	\draw[line width=1mm,pink,opacity=0.6] (ai1)--(ai);
	\draw[line width=1mm,pink,opacity=0.6] (ui)--(ai);

\end{tikzpicture}

\end{center}

\begin{flalign}
    est(pivot_{i+1}(q),\vi)&\le |pivot_{i+1}(q),\vii] \odot (\vii,\bii) \odot \bii \ra \bi \odot (\bi,\vi)| && \notag \\
    &\leq est(pivot_{i+1}(q),\vii) + |(\vii,\bii)| + |\bii \bi| + |(\bi,\vi)|&& \notag \\
    \intertext{Using \Cref{lem:basepart2}, $est(\vii, pivot_{i+1}(q)) \le |q~pivot_{i+1}(q)| + |q\bii| + 18(\lgn-(i+1))+1$. Thus, we get:}
    &\leq (|q~pivot_{i+1}(q)| + |q\bii| + 18(\lgn-(i+1))+1) + 1 + |\bii \bi| +1 && \notag \\
    &= |q~pivot_{i+1}(q)| + |q\bi| + 18(\lgn-(i+1))+3 && \notag
\end{flalign}

\end{proof}

By symmetry, we also get a better estimate between some select vertices in $A_{i+1}$ and $\ui$. We claim the following lemma.

\begin{lemma}
Let $q$ be a vertex closer to $t$ than $\bii$ on the $st$ path, i.e., $|tq| \le |t\bii|$.  Then, after  executing  Part 2 of \Cref{alg:main},

$$est(pivot_{i+1}(q),\ui) \le |q~pivot_{i+1}(q)|+|q\ai| + 18(\lgn-(i+1))+3$$
\end{lemma}

\subsection{Part 3: Estimate the distance between $\ui$ and $\vi$ }

 In Part 3 of our algorithm in the $i$-th iteration, we look at all pairs $(x,y)$ in  $A_i \times A_i$.
For each vertex \( w \) where either \( w \in \ball_{i+1}(x) \) or \( x = pivot_i(w) \), we estimate the distance between $x$ and $y$ using $w$ and $pivot_{i+1}(w)$. Let us first bound the running time of our algorithm. The running time of Part 1 and Part 2 (same as Part 1 and Part 2 of $\es$) is $\TL(n^2)$. The running time of Part 3 depends on number of possible vertices $w$ processed in \Cref{line:gen8} of \Cref{alg:main}. There are two cases:
\begin{enumerate}[noitemsep]
	\item $w \in ball_{i+1}(x)$:
	Using \Cref{lem:ballsize}, there are at most $\TL(2^{2^{i+1}})$ vertices in $\ball_{i+1}(x)$.  
	\item $x = pivot_i(w)$: Let us assume that there are $\ell_x$ vertices for which $x = pivot_i(\cdot)$. 
\end{enumerate}
Thus, the total number of vertices processed in \Cref{line:gen8} of \Cref{alg:main} is  $\TL\left(\ti{i+1} + \ell_
x \right)$. We can calculate the running time of  Part 3 of our algorithm as:
 \begin{align}
 	   \sum_{(x,y) \in A_i \times A_i} \TL\left( \ti{i+1} + \ell_x \right) &=\TL\left( \frac{n}{\ti{i}}\frac{n}{\ti{i}} \ti{i+1} + \sum_{(x,y) \in A_i \times A_i}\ell_x\right)\notag\\
 	  &=\TL\left( n^2 + \frac{n}{\ti{i}}\sum_{x \in A_i}\ell_x\right)\notag\\
 	  \intertext{Using the fact that each vertex has a unique $i$-th pivot, we have $\sum_{x \in A_i} \ell_x = n$. Thus, we get}
 	  &= \TL\left( n^2+ \frac{n^2}{\ti{i}}\right)\notag\\
 	  &=\TL(n^2)\notag
 \end{align}

Thus, the running time of all parts of our algorithm in the $i$-th iteration is $\TL(n^2)$. Since the algorithm runs for $\lgn$ iterations, the total running time of \Cref{alg:main} is $\TL(n^2)$. Let us now prove the statement in \Cref{lem:mainclaim} after the $i$-th iteration of our algorithm. We have alluded to the proof in the overview in \Cref{sec:overview}.  Assume $|s\aii| \le |t\bii|$. Otherwise, a symmetrical argument holds from the $t$ side. There are two cases:

	\begin{enumerate}
		\item $|\ai\aii| \le 5$

		 We can use \Cref{lem:basepart3} with $q = \ai$. Indeed, this is valid as $\ai$ lies closer to $s$ than $\aii$. Thus, we get:
		\begin{align}
			est(\vi,pivot_{i+1}(\ai)) &\le |\ai pivot_{i+1}(\ai)|+|\ai\bi| + 18(\lgn-(i+1))+3 \label{eq:9}
		\end{align}

		In Part 3 of our algorithm, we process all pairs in $A_i$. Thus, we process $(\ui,\vi)$ too. By definition, $\ui = pivot_i(a_i)$.  Thus, we would have estimated the distance between $\ui$ and $\vi$ using $x = \ui$, $y=\vi$ and $w = \ai$ in  \Cref{alg:main} (using the relation $x=pivot_i(w)$). This estimation represents the following path in the graph:
		$$(\ui,\ai) \odot [\ai,pivot_{i+1}(\ai)]\odot [pivot_{i+1}(\ai), \vi]$$ 
		Please see the figure below for a pictoral view of the path.
				
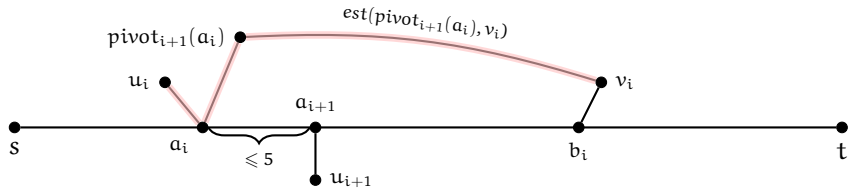
\begin{center}
		\begin{tikzpicture}
    	\node[draw, circle, fill=black, label=below:$s$, minimum size=4pt, inner sep=0pt] (s) at (-2,0) {};
    	\node[draw, circle, fill=black, label=below:$t$, minimum size=4pt, inner sep=0pt] (t) at (9,0) {};

    	\draw[thick] (s) -- (t);

    	\node[draw, circle, fill=black, label=below left:\scalebox{0.8}{$\ai$}, minimum size=4pt, inner sep=0pt] (ai) at (0.5,0) {};
    	\node[draw, circle, fill=black, label=above:\scalebox{0.8}{$\aii$}, minimum size=4pt, inner sep=0pt] (aii) at (2,0) {};
    	\node[draw, circle, fill=black, label=left:\scalebox{0.8}{$\ui$}, minimum size=4pt, inner sep=0pt] (ui) at (0,0.6) {};
    	\node[draw, circle, fill=black, label=right:\scalebox{0.8}{$\uii$}, minimum size=4pt, inner sep=0pt] (uii) at (2,-0.7) {};
    	\node[draw, circle, fill=black, label=left:\scalebox{0.8}{$pivot_{i+1}(\ai)$}, minimum size=4pt, inner sep=0pt] (ais) at (1,1.2) {};
    	\draw[thick] (ai) -- (ais);
		\draw[thick] (ai) -- (ui);
		\draw[thick] (aii) -- (uii);

    	\node[draw, circle, fill=black, label=below:\scalebox{0.8}{$\bi$}, minimum size=4pt, inner sep=0pt] (bi) at (5.5,0) {};
    	\node[draw, circle, fill=black, label=right:\scalebox{0.8}{$\vi$}, minimum size=4pt, inner sep=0pt] (vi) at (5.8,0.6) {};

		\draw[thick] (bi) -- (vi);
		\draw[thick] (ais) to[bend left=10] node[midway, above, sloped] {\scalebox{0.7}{$est(pivot_{i+1}(\ai),\vi$)}}(vi);
    	\draw[line width=1mm,pink,opacity=0.6] (ui)--(ai);
		\draw[line width=1mm,pink,opacity=0.6] (ais) to[bend left=10] (vi);
		\draw[line width=1mm,pink,opacity=0.6] (ai)--(ais);
		\draw [decorate, decoration={brace, amplitude=6pt, mirror}, thick] (ai) -- (aii) node[midway, below=6pt] {\scalebox{0.7}{$\le 5$}};
	\end{tikzpicture}

	\end{center}
		Thus, $est(\ui,\vi)$ will be updated as:

		\begin{align*}
			est(\ui,\vi) &\le est(\ui,\ai) + est(\ai,pivot_{i+1}(\ai))+ est(pivot_{i+1}(\ai), \vi)  \\
			\intertext{Using \Cref{lem:pivotandball}, we know the exact distance between a vertex and its pivots. Thus, we get:}
			&\le |(\ui,\ai)| + |\ai pivot_{i+1}(\ai)| + est(pivot_{i+1}(\ai), \vi)\\
			&= 1+ |\ai pivot_{i+1}(\ai)| + est(pivot_{i+1}(\ai),\vi) \\
			\intertext{Using \Cref{eq:9}, we get:}
			&\le1+|\ai pivot_{i+1}(\ai)|+(|\ai pivot_{i+1}(\ai)|+|\ai \bi| + 18(\log\log n-(i+1))+3)\\
			&\le2|\ai pivot_{i+1}(\ai)|+ |\ai\bi| + 18(\lgn-(i+1))+4\\
\intertext{Since $\uii$ is one of the candidate for the $i+1$-th pivot of $\ai$, $|a_i pivot_{i+1}(\ai)| \le |\ai\uii| \le |\ai\aii| + |(\aii,\uii)|$.
		Since $|\ai\aii| \le 5$, $|\ai pivot_{i+1}(\ai)| \le 6$. Thus, we get:}
			&\le 12+ |\ai\bi| + 18(\lgn-(i+1))+4\\
			& = |\ai\bi| + 18(\lgn-(i+1)) +16\\
			&\le  |\ai\bi| + 18(\llni)
		\end{align*}

	\item $|\ai\aii| \ge 6$

	For this case, we use our arguments from the overview  in \Cref{sec:overview}. Remember that  $|s\aii| \le |t\bii|$.  Let $p$ be the vertex at a distance $|s~\piis|$ from $s$ on the $st$ path. Using \Cref{lem:near}, since  
$|s~pivot_{i+1}(s)| - |s\aii| \le 3$, the distance between $p$ and $\aii$ on the $st$ path is at most 3. Using \Cref{lem:es}, there is a vertex, say $z$, at a distance 3 from $p$ or at a distance $\le 6$ from $\aii$ on $s\aii$ path such that  $z \in \ball_{i+1}(\ui)$. We now show that we can use \Cref{lem:basepart3} with $q = z$. This is true as:
	
	\begin{align*}
		|sz| & = |sp|-3\\
			 & =|s~\piis|-3\\
		\intertext{After \Cref{lem:near}, we have already argued that $s~\piis$ can be greater than $s\aii$ by at most 1. Thus, we get:}
			& \le |s\aii|+1-3\\
			& \le |s\aii|
	\end{align*}
	
Thus, $z$ is closer to $s$ than $\aii$ and we can use \Cref{lem:basepart3} with $q = z$. Thus, we get:

		\begin{align}
			est(pivot_{i+1}(z),\vi) 
			&\le |z~pivot_i(z)|+|z \bi| + 18(\lgn- (i+1))+3\label{eq:10}
		\end{align}
In Part 3 of our algorithm, we process the pair $(\ui,\vi)$. Using \Cref{lem:es}, $z \in \ball_{i+1}(\ui)$.  Thus, we would have estimated the distance between $\ui$ and $\vi$ using $x = \ui$, $y=\vi$ and $w = z$ in  \Cref{alg:main} (using the relation $w = \ball_{i+1}(x)$). This estimation represents the following path in the graph:

		$$[\ui,z]\ \odot [z,pivot_{i+1}(z)]\odot [pivot_{i+1}(z), \vi]$$ 
		Please see the figure below for a pictoral view of the path.

	
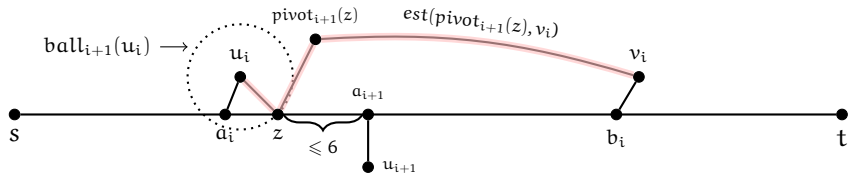
\begin{center}
	\begin{tikzpicture}
    \node[draw, circle, fill=black, label=below:$s$, minimum size=4pt, inner sep=0pt] (s) at (-2,0) {};
    \node[draw, circle, fill=black, label=below:$t$, minimum size=4pt, inner sep=0pt] (t) at (9,0) {};

    \draw[thick] (s) -- (t);

    \node[draw, circle, fill=black, label=above:\scalebox{0.8}{$\ui$}, minimum size=4pt, inner sep=0pt] (ui) at (1, 0.5) {};
    \node[draw, circle, fill=black, label=below:\scalebox{0.8}{$\ai$}, minimum size=4pt, inner sep=0pt] (ai) at (0.8, 0) {};
    \draw[thick] (ui) -- (ai);
	\node[draw, circle, fill=black, label=above:\scalebox{0.8}{\scalebox{0.8}{$\aii$}}, minimum size=4pt, inner sep=0pt] (aii) at (2.7,0) {};
 	\node[draw, circle, fill=black, label=right:\scalebox{0.8}{\scalebox{0.8}{$\uii$}}, minimum size=4pt, inner sep=0pt] (uii) at (2.7,-0.7) {};

    \draw[thick, dotted] (ui) circle[radius=0.7];

    \node at (-0.9, 0.9) (label) {\scalebox{0.8}{$\ball_{i+1}(\ui)$}};
    \draw[->] (label.east) -- (0.3, 0.9);

    \node[draw, circle, fill=black, label=above:\scalebox{0.8}{$\vi$}, minimum size=4pt, inner sep=0pt] (vi) at (6.3, 0.5) {};
	\node[draw, circle, fill=black, label=below:\scalebox{0.8}{$\bi$}, minimum size=4pt, inner sep=0pt] (bi) at (6,0) {};
    \draw[thick] (vi) -- (bi);

    \node[draw, circle, fill=black, label=below:\scalebox{0.8}{$z$}, minimum size=4pt, inner sep=0pt] (z) at (1.5, 0) {};
    \node[draw, circle, fill=black, label=above :\scalebox{0.65}{$pivot_{i+1}(z)$}, minimum size=4pt, inner sep=0pt] (zis) at (2, 1) {};
    \draw[thick] (y) -- (zis);
	\draw[thick] (ui) -- (z);
	\draw[thick] (aii) -- (uii);

	\draw[thick] (zis) to[bend left=10] node[midway, above, sloped] {\scalebox{0.7}{$est(pivot_{i+1}(z),\vi$)}}(vi);

    \draw[line width=1mm,pink,opacity=0.6] (z)--(ui);
    \draw[line width=1mm,pink,opacity=0.6] (z)--(zis);
	\draw[line width=1mm,pink,opacity=0.6] (zis) to[bend left=10] (vi);
	\draw [decorate, decoration={brace, amplitude=6pt, mirror}, thick] (z) -- (aii) node[midway, below=6pt] {\scalebox{0.7}{$\le 6$}};

	\end{tikzpicture}

	\end{center}
		Thus, $est(\ui,\vi)$ will be updated as:

		\begin{align}
			est(\ui,\vi) &\le est(\ui,z) + est(z,pivot_{i+1}(z))+ est(pivot_{i+1}(z), \vi)  \notag\\
			\intertext{ Using \Cref{lem:pivotandball}, we know the exact distance between a vertex and its pivots, as well as between the vertex and all other vertices in its $(i+1)$-th ball. Thus, we get:  
}
			&= |\ui z|+ |z~pivot_{i+1}(z)|+ est(pivot_{i+1}(z),\vi) \notag\\
			\intertext{Using  \Cref{eq:10}, we get:}
			&= |\ui z|+ |z~pivot_{i+1}(z)|+ (|z~pivot_{i+1}(z)| + |z \bi| + 18(\lgn- (i+1))+3) \notag\\
			&= |\ui z|+ 2|z~pivot_{i+1}(z)|  + |z \bi| + 18(\lgn- (i+1))+3 \notag\\
			\intertext{But $z$ is at a distance $\le 6$ from $\aii$. Using triangle inequality,  $|z \uii| \le |z\aii| + |(\aii,\uii)| \le  7$. This implies that $|z pivot_{i+1}(z)| \le 7$ as $\uii$ itself is one candidate for the $i+1$-th pivot of $z$. Thus, we have: }
			&\le  |\ui z|+14+|z \bi| + 18(\lgn- (i+1))+3 \notag\\
			&= |\ui z|+ |z \bi| + 18(\lgn- (i+1))+17 \notag\\
			\intertext{Using triangle inequality for $|\ui z|$, we get:}\notag
			&\le  (|(\ui,\ai)|+ |\ai z|)+ (|z\bi|+18(\lgn-(i+1))+ 17 \notag\\
			&\le  (1+ |\ai z|)+ (|z\bi|+18(\lgn-(i+1))+ 17 \notag\\
			&=|\ai \bi| + 18(\lgn-(i+1))+18\notag \\
			& = |\ai\bi| + 18(\llni) \notag
		\end{align}
	\end{enumerate}
	
	This completes the proof of \Cref{lem:mainclaim}. The reader can verify that all our algorithms, \Cref{alg:base}, \Cref{alg:main}, and $\es$, run in $\TL(n^2)$ time. Thus, we have proven \Cref{thm:main} when $k = \log n$.

\section{A general bound}
\label{sec:genk}

In this section, we prove \Cref{thm:main}. To achieve this, we slightly modify the algorithm from the previous section. Specifically,  
we run \Cref{alg:main} from iteration $\lgn-2$ to $i+1$ and then stop. We determine the exact value of $i$ later. After that, we run the following algorithm:

\begin{algorithm}[H]
\caption{Final Step in the algorithm for general bound}
\label{alg:genk}

	{\color{blue}\tcc{Part 1}}
	\ForEach{$w$ in $A_{i+1}$}
	{
		\uf$(w,i+1)$\\
	}

	{\color{blue}\tcc{Part 2}}
	\ForEach{$w$ in $A_{i+1}$}
	{
		\uf$(w,i+1)$\\
	}

{\color{blue}\tcc{Part 3}}
	\ForEach{$x,y$ in $A_0 \times A_0$}
	{
		\ForEach{ $w \in ball_{i+1}(x) \cup \{x\}$}
		{\label{line:genk8}

				$
				est(x,y) = min
				\begin{cases}
					est(x,w)+est(w,pivot_{i+1}(w)) + est(pivot_{i+1}(w),y)\\
					est(x,y)
				\end{cases} \label{line:gen9}
				$
		}

	}

\end{algorithm}

The above algorithm is almost identical to \Cref{alg:main}, except for Part 3, where we operate on $A_0 \times A_0$ instead of $A_{i} \times A_{i}$.  Next, we review each part of the algorithm. As in \Cref{sec:part1}, Part 1 provides an estimate between some vertices of $A_{i+1}$ and $\vii$. Now, we examine Part 2.

\subsection{Part 2: Estimate the distance between some vertices in $A_{i+1}$ and $t$}
In Part 2 of our algorithm, we will give a new estimate between   some vertices in $A_{i+1}$   and $t$. To this end, we prove the following lemma. 

\begin{lemma} \label{lem:genpart2}
Let $q$ be a vertex closer to $s$ than $\aii$ on the $st$ path, i.e., $|sq| \leq |s\aii|$.  Then, after  executing Part 2 of \Cref{alg:genk} ,

$$est(pivot_{i+1}(q),t) \le |q~pivot_{i+1}(q)|+|qt| + 18(\lgn-(i+1))+2$$
\end{lemma}

\begin{proof}

In Part 2 of \Cref{alg:genk}, we call the procedure $\uf$ all vertices of $A_{i+1}$. So, we also call the procedure $\uf(pivot_{i+1}(q),i+1)$. In this procedure, we construct a graph $H_{pivot_{i+1}(q)}$. We claim that there is a path from $pivot_{i+1}(q)$ to $t$ in $H_{pivot_{i+1}(q)}$ as follows:

$$[pivot_{i+1}(q),\vii] \odot (\vii,\bii) \odot \bii \ra t $$
The first weighted edge is added in \Cref{line:uf1} of $\uf(pivot_{i+1}(q),i+1)$. The second  edge is added in \Cref{line:uf3}. Using \Cref{lem:onpath}, each edge in $\bii t$ has degree $\TL(\ti{i+1})$ and is added in \Cref{line:uf2}.    Thus, we update the estimate between  $pivot_{i+1}(q)$ and $t$ as follows:

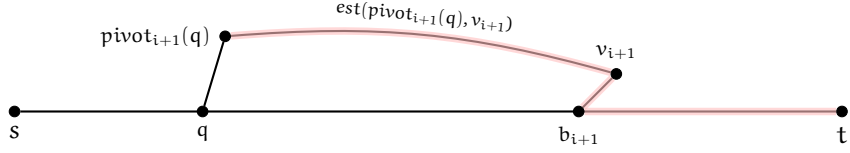
\begin{center}

\begin{tikzpicture}
    \node[draw, circle, fill=black, label=below:$s$, minimum size=4pt, inner sep=0pt] (s) at (-2,0) {};
    \node[draw, circle, fill=black, label=below:$t$, minimum size=4pt, inner sep=0pt] (t) at (9,0) {};

    \draw[thick] (s) -- (t);

    \node[draw, circle, fill=black, label=below:\scalebox{0.8}{$q$}, minimum size=4pt, inner sep=0pt] (x) at (0.5,0) {};
    \node[draw, circle, fill=black, label=left:\scalebox{0.8}{$pivot_{i+1}(q)$}, minimum size=4pt, inner sep=0pt] (pix) at (0.8,1) {};
    \draw[thick] (x) -- (pix);

    \node[draw, circle, fill=black, label=above:\scalebox{0.8}{$\vii$}, minimum size=4pt, inner sep=0pt] (ui) at (6, 0.5) {};
	    \node[draw, circle, fill=black, label=below:\scalebox{0.8}{$\bii$}, minimum size=4pt, inner sep=0pt] (ai) at (5.5,0) {};
    \draw[thick] (ui) -- (ai);

	    \draw[ thick] (ui) to[bend left=-10] node[midway, above, sloped] {\scalebox{0.7}{$est(pivot_{i+1}(q),\vii$)}} (pix);
	    \draw[line width=1mm,pink,opacity=0.6] (ui) to[bend left=-10]  (pix);

	\draw[line width=1mm,pink,opacity=0.6] (t)--(ai);
	\draw[line width=1mm,pink,opacity=0.6] (ui)--(ai);

\end{tikzpicture}

\end{center}

\begin{flalign}
    est(pivot_{i+1}(q),t)&\le |pivot_{i+1}(q),\vii] \odot (\vii,\bii) \odot \bii \ra t | && \notag \\
    &\leq est(pivot_{i+1}(q),\vii) + |(\vii,\bii)| + |\bii t|&& \notag \\
    \intertext{After Part 1 of \Cref{alg:genk}, \Cref{lem:basepart2} holds. Thus, $est(\vii, pivot_{i+1}(q)) \le |q~pivot_{i+1}(q)| + |q\bii| + 18(\lgn-(i+1))+1$. Thus, we get:}
    &\leq (|q~pivot_{i+1}(q)| + |q\bii| + 18(\lgn-(i+1))+1)+1 + |\bii t|  && \notag \\
    &= |q~pivot_{i+1}(q)| + |qt| + 18(\lgn-(i+1))+2 && \notag
\end{flalign}

\end{proof}

\subsection{Part 3: Estimate the distance between $s$ and $t$}

In Part 3 of \Cref{alg:genk}, we estimate the distance between $s $ and $t$ using vertices in $\ball_{i+1}(s)$. We now analyze the estimates produced by our algorithm. Assume that $|s\aii| \le |t\bii|$; otherwise, we apply a symmetrical argument from the $t$ side.
 There are two cases:
 \begin{enumerate}
 	\item $|s\aii| \le 3$
 	
 	We can apply \Cref{lem:genpart2} with $q=s$. Thus, we get 
 	\begin{align}
 		est(pivot_{i+1}(s),t) &\le |s~pivot_{i+1}(s)|+|s t| + 18(\lgn- (i+1))+2 \label{eq:11}
 	\end{align}
 	
 	In Part 3 of our algorithm, we process the pair $(s,t)$. But, we  have estimated the distance between $s$ and $t$ using $x = s$, $y=t$ and $w = s$ in  \Cref{alg:genk}. This estimation represents the following path in the graph:
 	$$ [s,pivot_{i+1}(s)]\odot [pivot_{i+1}(s), t]$$ 

			Thus, $est(s,t)$ will be updated as:
 	\begin{align*}
 		est(s,t) &\le est(s,pivot_{i+1}(s))+ est(pivot_{i+1}(s),t)
 		\intertext{Using \Cref{eq:11}, we get:}
 		& \le |s~pivot_{i+
 		1}(s)| + (|s~pivot_{i+1}(s)|+|s t| + 18(\lgn- (i+1))+2)\\
 		&=2|s~pivot_{i+1}(s)|+|s t| + 18(\lgn- (i+1))+2
 		\intertext{ Using triangle inequality,  $|s \uii| \le |s\aii| + |(\aii,\uii)| \le  4$. This implies that $|s~pivot_{i+1}(s)| \le 4$ as $\uii$ itself is one candidate for the $i+1$-th pivot of $s$. Thus, we have: }
 		&\le 8 +|s t| + 18(\lgn- (i+1))+2\\
 		&\le |s t| + 18(\lgn- (i+1))+10\\
 		&\le |s t| + 18(\lgn- i)
 	\end{align*}
 	
 	\item $|s\aii| \ge 4$
 	
 	    By \Cref{lem:near}, we have  
$|s~pivot_{i+1}(s)| - |s\aii| \le 3$. Let $p$ be the vertex at distance $|s~pivot_{i+1}(s)|$ from $s$ on the $st$ path. Since  
$
|s~pivot_{i+1}(s)| - |s\aii| \le 3,
$
the distance between $p$ and $\aii$ on the $st$ path is at most 3.  
Let $z$ be the vertex just before $p$ on the $st$ path (see the figure below). Since $z$ is strictly closer to $s$ than $pivot_{i+1}(s)$, it must lie in $\ball_{i+1}(s)$. We now apply \Cref{lem:genpart2} for $z$. Indeed, as in the previous section we can show that $z$ is closer to $s$ than $\aii$. Thus, we have:

		\begin{align}
			est(pivot_{i+1}(z),t) \le |z~pivot_{i+1}(z)|+|z t| + 18(\lgn- (i+1))+2\label{eq:12}
		\end{align}
		
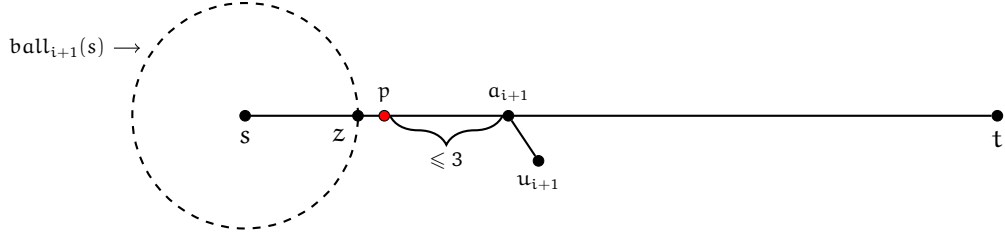
\begin{center}

\begin{tikzpicture}
    \node[draw, circle, fill=black, label=below:$s$, minimum size=4pt, inner sep=0pt] (s) at (0,0) {};
    \node[draw, circle, fill=black, label=below:$t$, minimum size=4pt, inner sep=0pt] (t) at (10,0) {};
    \node[draw, circle, fill=black, label=above:\scalebox{0.8}{$\aii$}, minimum size=4pt, inner sep=0pt] (aii) at (3.5,0) {};
    \node[draw, circle, fill=black, label=below:\scalebox{0.8}{$\uii$}, minimum size=4pt, inner sep=0pt] (vii) at (3.9,-0.6) {};

    \draw[thick] (s) -- (t);
    \draw[thick] (aii) -- (vii);

    \draw[thick, dashed] (s) circle[radius=1.5];

    \node[draw, circle, fill=red, label=above :\scalebox{0.8}{$p$}, minimum size=4pt, inner sep=0pt] (p) at (1.85,0) {};
	\node[draw, circle, fill=black, label=below left:$z$, minimum size=4pt, inner sep=0pt] (z) at (1.5,0) {};
	    \node at (-2.5, 0.9) (label) {\scalebox{0.8}{$\ball_{i+1}(s)$}};
    \draw[->] (label.east) -- (-1.4, 0.9);

    \draw [decorate, decoration={brace, amplitude=11pt, mirror}, thick] (p) -- (aii)
        node[midway, below=10pt] {\scalebox{0.8}{$\le 3$}};

\end{tikzpicture}
\end{center}
	
	In Part 3 of our algorithm, we process the pair $(s,t)$. Since $z \in \ball_{i+1}(s)$, we would have estimated the distance between $s$ and $t$ using $x = s$, $y=t$ and $w = z$ in  \Cref{alg:genk}. This estimation represents the following path in the graph:

		$$[s,z]\ \odot [z,pivot_{i+1}(z)]\odot [pivot_{i+1}(z)], t]$$ 

			Thus, $est(s,t)$ will be updated as:

		\begin{align}
			est(s,t) &\le est(s,z) +est(z,pivot_{i+1}(z))+est(pivot_{i+1}(z), t)  \notag\\
			\intertext{ Using \Cref{lem:pivotandball}, we know the exact distance between a vertex and its pivots, as well as between the vertex and all other vertices in its $(i+1)$-th ball. Thus, we get:  
}
			&= |s z|+ |z~pivot_{i+1}(z)|+ est(pivot_{i+1}(z),t) \notag\\
			\intertext{Using  \Cref{eq:12}, we get:}
			&= |s z|+ |z~pivot_{i+1}(z)|+ (|z~pivot_{i+1}(z)| + |z t| + 18(\lgn- (i+1))+2) \notag\\
			&= |s t|+ 2|z~pivot_{i+1}(z)|   + 18(\lgn- (i+1))+2 \notag\\
			\intertext{But,  $z$ is at a distance $\le 4$ from $\aii$ (see the figure above). Using triangle inequality, $|z\uii| \le |z\aii| + |(\aii,\uii)| \le 5$. This implies that $|z~pivot_{i+1}(z)| \le 5$ as $\uii$ itself is one candidate for the $i+1$-th pivot of $z$. Thus, we have: }
			&\le  |s t|+10 + 18(\lgn- (i+1))+2 \notag\\
			&\le  |s t| + 18(\lgn- (i+1))+12 \notag\\
			&\le  |s t| + 18(\lgn- i)\notag
		\end{align}

 \end{enumerate}
 Thus, we get $est(s,t) \le |st| + 18(\lgn-i)$ in both cases. Putting $i = \llnk-1$ in this equation, we get $est(s,t) \le |st| + 18 (\log k+1)$. If $|st| \ge 18 (\log k+1)$, then we get a 2-approximation of $st$ path. 

We now bound the running time of our algorithm.  Parts 1 and 2, discussed in the previous section, run in $\TL(n^2)$ time. Now, we analyze Part 3. By \Cref{lem:ballsize}, the number of vertices in $\ball_{i+1}(\cdot)$ is $\TL(\ti{i+1})$. Thus, the running time of Part 3 is  
\begin{align*}
\sum_{x,y \in A_0 \times A_0} \TL(\ti{i+1}) &= \TL(n^2 \ti{i+1}).
\intertext{The rest of the algorithm runs in $\TL(n^2)$ time. Thus, the total running time of our algorithm is dominated by the running time of Part 3. Putting $i = \lgn-\log k-1$, we get:}
& = \TL(n^2 \ti{\llnk})\\
					  & = \TL(n^2 2^{\frac{\log n}{k}})\\
					  &= \TL(n^2 \times n^{1/k})\\
					  & = \TL(n^{2+\frac{1}{k}})
\end{align*}

This completes the proof of our main  \Cref{thm:main}.

\bibliographystyle{alpha}
\bibliography{../../jabref.bib}
\appendix
\section{Proof of \Cref{lem:prob}}
We will prove this using induction on $i$. For the base case, when $i=1$, the lemma follows from the definition of $A_{1}$. Let us now assume using induction hypothesis, that $Pr[v \in A_{i-1}] = \frac{1}{\ti{i-1}}$, for $i \ge 1$.

\begin{equation*}
\begin{split}
	Pr[v \in A_i] &= Pr[\text{$v$ is chosen in $A_i | $ $v \in A_{i-1}]~Pr[v \in A_{i-1}]$ }\\
	&= \frac{1}{\ti{i-1}}  \times \frac{1}{\ti{i-1}}\\
	&=\frac{1}{\ti{i}}
\end{split}
\end{equation*}

\section{Proof of \Cref{lem:ballsize}}

Let us assume for contradiction that $\ball_i(s)$ contains $\ge \TL(\ti{i})$ vertices. 
Order these vertices by increasing distance from $s$, breaking ties arbitrarily. Define $S$ as the first $c \ti{i} \log n$ vertices in this order. Each vertex $v \in S$ belongs to $A_i$ with probability $\frac{1}{\ti{i}}$. Thus, the probability that no vertex in $S$ is in $A_i$ is at most:

\[
\left(1 - \frac{1}{\ti{i}}\right)^{c \ti{i} \log n} \le e^{-c \log n} = \frac{1}{n^c}.
\]

Thus, with a high probability, a vertex in $\ball_i(s)$ is in $A_i$. But this is a contradiction as by definition $\ball_i(s)$ does not contain any vertex from $A_i$. Thus our assumption is wrong, or with a high probability, the size of $\ball_i(s)$ is $\le \TL(\ti{i})$.
By a union bound over all vertices $s \in G$ and all $i$, with high probability, $\ball_i(s)$ is of size $\le \TL(\ti{i})$.

\section{Proof of \Cref{lem:pivotandball}}
For each $w \in A_i$, we construct a graph $H_w$ containing:
\begin{enumerate}
    \item all neighbors of $w$,
    \item all edges with degree at most $\TL(\ti{i})$.
\end{enumerate}

The first point ensures that the last edge on the $s\pis$ path, if of high degree, is included. The second point accounts for all other edges on the $s\pis$ path.  The reader can verify that the size of $H_w$ is $\TL(n \ti{i})$. We perform a BFS from $w$ in $H_w$, storing the shortest path to all vertices in $G$. The total runtime is
\[
\sum_{w \in A_i} \TL(n \ti{i}) = \TL\left(\frac{n}{\ti{i}} \times n \ti{i}\right) = \TL(n^2).
\]
Finally, each vertex $s$ selects the vertex in $A_i$ reporting the shortest path as $\pis$, breaking ties arbitrarily.

We now show how to compute $\ball_i(s)$. We perform a BFS from $s$, processing only vertices at distance $\leq |s\pis| - 2$. By processing, we mean that we examine all the neighbors of these vertices in the BFS algorithm. Since $\ball_i(s)$ includes all vertices within distance $|s\pis| - 1$ of $s$, this process correctly discover all vertices in $\ball_i(s)$.  By \Cref{lem:ballsize}, size of $\ball_i(s)$ is $\TL(\ti{i})$, and by \Cref{lem:degreeinball}, each processed vertex has degree $\TL(\ti{i})$. Thus, the BFS runs in

\[
\TL(\ti{i} \times \ti{i}) = \TL(\ti{i+1}).
\]
Summing over all vertices $s \in G$ and all $i$, the total runtime is
\[
\sum_{s \in G} \sum_{i=0}^{\log\log n-1} \TL(\ti{i+1}) = \sum_{s \in G} \TL(n) = \TL(n^2).
\]

\end{document}